\documentclass[11pt]{article}
\usepackage[plain]{fullpage}
\usepackage{amsmath}
\usepackage{graphicx}
\usepackage{amssymb}
\usepackage{url}
\usepackage{bbm}
\usepackage{epsfig}
\usepackage{sectsty}
\usepackage{nicefrac}
\usepackage{enumitem}
\usepackage{algorithm}

\usepackage{array}
\usepackage{tkz-graph}
\usepackage{tikz}
\usetikzlibrary{decorations.pathreplacing}
\usepackage{algpseudocode}

\tikzstyle{b_vertex}=[circle,fill=black!100,text=white,inner sep=0.8mm,draw]
\tikzstyle{w_vertex}=[circle,fill=white!100,text=black,inner sep=0.8mm,draw]
\tikzstyle{e_vertex}=[circle,fill=white!100,text=black,inner sep=1.8mm,draw]
\tikzstyle{point}=[circle,fill=black,inner sep=0.1mm]
\tikzstyle{path_edge}=[thick]

\newtheorem{lemma}{Lemma}[section]
\newtheorem{theorem}{Theorem}[section]
\newtheorem{definition}{Definition}[section]

\newtheorem{rmks}{Remarks}[section]

\newtheorem{conjecture}{Conjecture}

\newenvironment{proof}{{\bf Proof}:\ }%
   {~\ \hfill $\Box$\vspace{0,1cm}}

\allsectionsfont{\hspace{-\parindent}\textsf}

\begin{document}

\newcommand\thetitle{Dominating induced matchings \\ in graphs
containing no long claw}
\title{\textbf{\thetitle}}
\author{ 
\Large \bf{Alain Hertz}\footnote{Corresponding author: email  alain.hertz@gerad.ca;  tel. +1-514 340 6053; fax +1-514 340 5665.}\\
\normalsize {Polytechnique Montr\'eal and GERAD, Canada}\vspace{0.1cm}\\
 \and
\Large \bf{Vadim Lozin} \\
\normalsize {University of Warwick, United Kingdom}\vspace{0.1cm}\\
 \and
\Large \bf{Bernard Ries}\\
\normalsize {PSL, Universit\'e Paris Dauphine and CNRS, France}\vspace{0.1cm}\\
\and
\Large \bf{Victor Zamaraev}\\
\normalsize {University of Warwick, United Kingdom}
\and
\Large \bf{Dominique de Werra} \\
\normalsize {\'Ecole Polytechnique F\'ed\'erale de Lausanne, Switzerland}\vspace{0.1cm}\\
}

%
%
%
%
%
%
%
%

\thispagestyle{empty}
\maketitle
\begin{abstract}
\noindent 
An induced matching $M$ in a graph $G$ is dominating if every edge not in $M$ shares exactly one vertex 
with an edge in $M$. The {\sc dominating induced matching} problem 
(also known as {\sc efficient edge domination}) asks whether a graph $G$ contains a dominating 
induced matching. This problem is generally NP-complete, but polynomial-time solvable
for graphs with some special properties. In particular, it is solvable in polynomial time 
for claw-free graphs. In the present paper, we study this problem for graphs containing no long claw,
i.e. no induced subgraph obtained from the claw by subdividing each of its edges exactly once.
To solve the problem in this class, we reduce it to the following question: given a graph 
$G$ and a subset of its vertices, does $G$ contain a matching saturating all vertices 
of the subset? We show that this question can be answered in polynomial time, thus providing 
a polynomial-time algorithm to solve the {\sc dominating induced matching} problem for graphs 
containing no long claw.     
\end{abstract}



\section{Introduction}\label{sec:intro}


In this paper, we study the problem that appeared in the literature under various names, such as 
{\sc dominating induced matching} \cite{P7,CL09,CKL11,Kor09,skew,fast} or 
{\sc efficient edge domination} \cite{hole-free,CCDS08,GSSH93,LKT02,LT98},
and has several equivalent formulations. One of them, which is used in this paper, asks whether 
the vertices of a graph can be partitioned into two subsets $B$ and $W$ so that $B$ induces a graph of vertex degree 1 
(also known as an induced matching) and $W$ induces a graph of vertex degree 0 (i.e. an independent set). 
Throughout the paper, we call the vertices of $B$ black and the vertices of $W$ white.
This problem finds applications in various fields, such as parallel resource allocation of parallel
processing systems \cite{LS88}, encoding theory and network routing \cite{GSSH93} and has 
relations to some other algorithmic graph problems, such as {\sc 3-colorability} and
{\sc maximum induced matching}. In particular, it is not difficult to see that 
every graph that can be partitioned into an induced matching and a stable set is 3-colorable. Also, in \cite{CCDS08} 
it was shown that if a graph admits such a partition, then the black vertices form 
an induced matching of maximum size. Notice that a graph is called polar if its vertex set can be partitioned into a subset
$\mathcal{K}$ of disjoint cliques and a subset $\mathcal{I}$ of independent sets with complete links between them \cite{TC85}. It follows that a graph $G$ has a dominating induced matching if and only if $G$ is a polar graph in which all cliques of $\mathcal{K}$ have 
size 2 and $\mathcal{I}$ consists of exactly one independent set.

From an algorithmic point of view, the {\sc dominating induced matching} problem is difficult, 
i.e. it is NP-complete \cite{GSSH93}. Moreover, it remains difficult under substantial restrictions,
for instance, for planar bipartite graphs \cite{LKT02} or $d$-regular graphs for arbitrary $d\ge 3$ \cite{CCDS08}.
On the other hand, for some special graph classes, such as hole-free graphs \cite{hole-free}, 
claw-free graphs \cite{CKL11} or $P_7$-free graphs \cite{P7}, the problem can be solved in polynomial time.

For classes defined by finitely many forbidden induced subgraphs, there is an important 
{\it necessary} condition for polynomial-time solvability of the problem given in \cite{CKL11}.
To state this condition, let us denote by $\cal S$ the class of graphs every connected component 
of which corresponds to a graph $S_{i,j,k}$ represented in Figure~\ref{fig:ST}. 
\begin{figure}[!ht]
\begin{center}
\scalebox{0.20}{\includegraphics{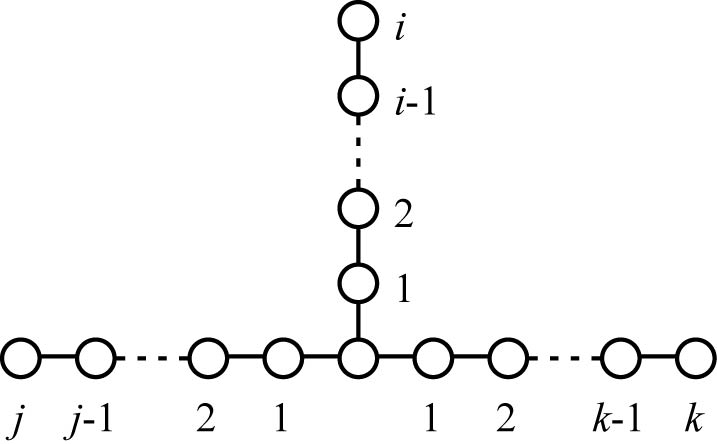}}\caption{The graph $S_{i,j,k}$.}
\label{fig:ST}
\end{center}
\end{figure}

\vspace{-0.4cm}\begin{theorem}{\rm \cite{CKL11}}\label{thm:1}
Let $M$ be a finite set of graphs. Unless $P=NP$, 
the  {\sc dominating induced matching} problem is polynomial-time solvable in the class of $M$-free graphs 
only if $M$ contains a graph from $\cal S$.
\end{theorem}

We believe that this necessary condition is also sufficient and formally state this as a conjecture.

\begin{conjecture}\label{con:1}   
Let $M$ be a finite set of graphs. Unless $P=NP$, 
the  {\sc dominating induced matching} problem is polynomial-time solvable in the class of $M$-free graphs if and only if
$M$ contains a graph from $\cal S$.
\end{conjecture}

Proving (or disproving) this conjecture is a very challenging task. To prove it, one has to show 
that the problem becomes polynomial-time solvable by forbidding {\it any} graph from $\cal S$.
However, so far, the conjecture has only been verified for a few forbidden graphs that belong to $\cal S$,
and only two of these classes are maximal: $S_{1,2,3}$-free graphs \cite{skew}
and $P_7$-free graphs \cite{P7} (note that $P_7=S_{0,3,3}$). In the present paper, we extend this 
short list of positive results by one more class where the problem can be solved in polynomial time,
namely, the class of $S_{2,2,2}$-free graphs. Since $S_{2,2,2}$ is obtained from the claw ($S_{1,1,1}$)
by subdividing each of its edges exactly once, we call $S_{2,2,2}$ 
a {\it long claw}.

To solve the problem for graphs containing no long claw, we apply a number of transformations and reductions
that eventually reduce the problem to the following question: given a graph 
$G$ and a subset of its vertices, does $G$ contain a matching saturating all vertices 
of the subset? We show that this question can be answered in polynomial time. As a result, we prove that
the {\sc dominating induced matching} problem for graphs containing no long claw can also be solved in polynomial time.

\medskip
The organization of the paper is as follows. In the rest of this section, we introduce basic terminology and notation.
In Sections~\ref{sec:DIM}, \ref{sec:REDUCTION} and \ref{sec:TRANSFORMATION} we describe various tools (reductions and transformations)
simplifying the problem. In Section~\ref{sec:Irr} we apply these tools in order to reduce the problem from an arbitrary 
$S_{2,2,2}$-free graph $G$ to a graph of particular structure, which we call irreducible.  
Finally, in Section~\ref{sec:solution} we show how to solve the problem for irreducible graphs via finding matchings saturating specified vertices.     
In Section~\ref{sec:conclusion}, we conclude the paper with a number of open problems. 

\medskip  
Let $G=(V,E)$ be a graph. If $v\in V$, then $N_G(v)$ is the \textit{neighborhood} of $v$ in $G$, i.e. the set of vertices of $G$ adjacent to $v$,
and $d_G(v)$ is the \textit{degree} of $v$ in $G$, i.e. $d_G(v)=|N_G(v)|$. 

An independent set in $G$ is a subset of pairwise nonadjacent vertices.
For a subset $U\subseteq V$, we denote by $G[U]$ the subgraph of $G$ induced by vertices of $U$.
If a graph $G$ does not contain induced subgraphs isomorphic to a graph $H$, we say that $G$ is \textit{$H$-free} and 
call $H$ a \textit{forbidden induced subgraph} for $G$. 
As usual, $K_n$ is the complete graph on $n$ vertices, and $C_n$ (resp. $P_n$) is the chordless cycle (resp. path) on $n$ vertices. 
A \textit{diamond} and a \textit{butterfly} are two special graphs represented in Figure~\ref{fig:DB}.

\begin{figure}[!ht]
\begin{center}
\scalebox{0.18}{\includegraphics{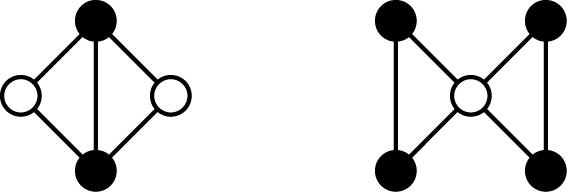}}\caption{A diamond (left) and a butterfly (right).}
\label{fig:DB}
\end{center}
\end{figure}


\section{Precoloring, propagation rules and cleaning}
\label{sec:DIM}


In order to solve our problem for a graph $G$, we will assign either color black or  color white to the vertices of $G$, 
and the assignment of one of the two colors to each vertex of $G$ is called a \textit{complete coloring} of $G$. 
If only some vertices of $G$ have been assigned a color, the coloring is said to be \textit{partial}. 
A partial coloring is \textit{feasible}, if no two adjacent vertices are white and every black vertex has at most one black neighbor. 
A complete coloring is \textit{feasible}, if no two adjacent vertices are white and every black vertex has exactly one black neighbor. 
Thus, a graph $G$ has a dominating induced matching if and only if $G$ admits a feasible complete coloring. 
Given a feasible partial coloring $\gamma$ of $G$, we say that it is {\it completable} if it can be extended to a feasible complete coloring of $G$, 
the latter one being called a {\it $\gamma$-completion}. Also, for a feasible partial coloring $\gamma$, 
we denote by $\gamma(v)$ the color of vertex $v$, by $B_{\gamma}$ the set of black vertices and by $W_{\gamma}$ the set of white vertices.

Let $\gamma$ be a feasible partial coloring of a graph $G$, and let $G'$ be the graph obtained from $G$ by removing 
all white vertices as well as all pairs of adjacent black vertices. The restriction $\delta$ of $\gamma$ to $G'$ is a feasible partial 
coloring of $G'$ where some of its vertices are forced to be black and form an independent set.
Clearly, $\gamma$ is completable if and only if $\delta$ is. The construction of $G'$ and $\delta$ is called a {\it cleaning}.

As shown in the following lemma, there are situations where some vertices of a graph $G$ must have the same 
color or necessarily have different colors in all feasible complete colorings of $G$.

\vspace{0.2cm}\begin{lemma}
\label{lem:samedifferent}
Let $\gamma$ be a feasible complete coloring of a graph $G=(V,E)$.
\begin{itemize}
\vspace{-0.3cm}\item[(i)] If $G$ contains $C_4$ with edge set $\{v_1v_2,v_2v_3,v_3v_4,$ $v_1v_4\}$, 
then $\gamma(v_1)=\gamma(v_3)\neq \gamma(v_2)=\gamma(v_4)$.
\vspace{-0.3cm}\item[(ii)] If $G$ contains a triangle with vertex set $\{x,y,z\}$ such that $x$ has a neighbor $u$ 
which is not adjacent to $y,z$, then $\gamma(x)\neq \gamma(u)$.
\end{itemize}
\end{lemma}

\vspace{-0.5cm}\begin{proof}
\begin{itemize}
\vspace{-0.3cm}\item[(i)] Clearly, at least one of $v_1,v_2,v_3,v_4$ is white, say $v_1$. 
Then both $v_2$ and $v_4$ are black, which means that $v_3$ must be white since it has two black neighbors.
\vspace{-0.3cm}\item[(ii)] Clearly, at least one of $y,z$ is black. Hence, if $x$ is black, then $u$ is white. 
Since $x$ and $u$ cannot be both white, we  conclude that $\gamma(x)\neq \gamma(u)$.
\end{itemize}
\vspace{-0.8cm}\end{proof}

We now describe several situations where the color of a vertex $v$ can be fixed because if there is a feasible complete coloring of $G$, then there is at least one in which $v$ has such a color. The graphs $F_i$, $i=1,\ldots,10$ we refer to in the following lemma are shown in Figure \ref{figF123456789}.
\begin{figure}
\begin{center}
\scalebox{0.15}{\includegraphics{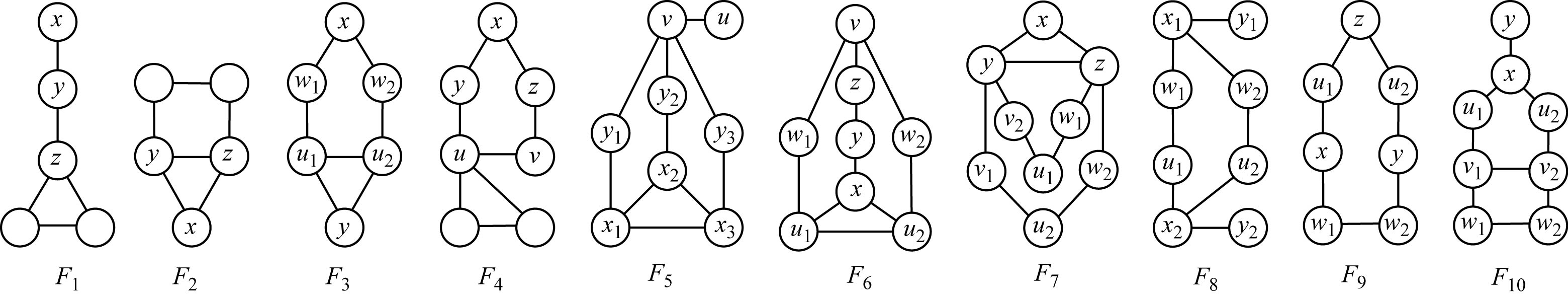}}\caption{The graphs $F_1,\ldots, F_{10}$.}
\label{figF123456789}
\end{center}
\end{figure}

\vspace{0.1cm}\begin{lemma}
\label{lem:simpleforcing}
Let $\gamma$ be a feasible partial coloring of a graph $G=(V,E)$. If $\gamma$ is completable, then the following rules are valid for obtaining 
a $\gamma$-completion. 
\begin{itemize}
\vspace{-0.3cm}\item[(a)] An isolated vertex must be white, and the neighbor of a vertex of degree 1 must be black.
\vspace{-0.3cm}\item[(b)] If two non-adjacent vertices in $G$ are in $B_{\gamma}$, then all their common neighbors must be  white.
\vspace{-0.3cm}\item[(c)] If two triangles in $G$ share a single vertex, then this vertex must be white. 
\vspace{-0.3cm}\item[(d)] If two triangles in $G$ share two vertices, then both of these vertices must be black.
\vspace{-0.3cm}\item[(e)] If a vertex $u$ of $G$ has $k>1$ neighbors of degree 1, then at least $k-1$ of these neighbors are not colored black, and color white can be assigned to them.
\vspace{-0.3cm}\item[(f)] Suppose $G$ contains $P_4$ with edge set $\{v_1v_2,v_2v_3,v_3v_4\}$. If $d_G(v_3)=2$  and $v_1$ is black, then $v_4$ must be black.
\vspace{-0.3cm}\item[(g)] Suppose $G$ contains $F_1$ as an induced subgraph. If $d_G(y)=2$, then $x$ must be black.
\vspace{-0.3cm}\item[(h)] If $G$ contains $F_2$ as an induced subgraph, then $x$ must be black.
\vspace{-0.3cm}\item[(i)] If $G$ contains $F_3$ as an induced subgraph, then $x$ must be black. 
Moreover, 
\begin{itemize}
\vspace{-0.4cm}\item if $d_G(x)=2$, then $y$ must be black;
\vspace{-0.2cm}\item if $y$ is black, then all neighbors $z\neq w_1,w_2$ of $x$ must be white;
\vspace{-0.2cm}\item if $d_G(u_1)=d_G(u_2)=3$ and $d_G(w_1)=d_G(w_2)=2$ and if these vertices are not yet colored by $\gamma$, then color white can be assigned to $w_1$.
\end{itemize}
\vspace{-0.4cm}\item[(j)] Suppose $G$ contains $F_4$ as an induced subgraph. If $d_G(y)=d_G(z)=2$, then $x$ must be black and all neighbors $w\neq y,z$  of $x$ must be white. 
\vspace{-0.3cm}\item[(k)] Suppose $G$ contains $F_5$ as an induced subgraph. If $d_G(y_i)=d_G(y_2)=d_G(y_2)=2$, then $u$ must be white.
\vspace{-0.3cm}\item[($\ell$)] Suppose $G$ contains $C_4$. If $d_G(v)=2$ for some vertex $v$ of this $C_4$, then $v$ must be white.
\vspace{-0.3cm}\item[(m)] Suppose $G$ contains $F_6$ as an induced subgraph. If $v$ is black, then $x$ must be white.
\vspace{-0.3cm}\item[(n)] Suppose $G$ contains $F_7$ as an induced subgraph. If $d_G(y)=d_G(z)=4$, $d_G(v_i)=d_G(w_i)=2$, $i=1,2$ and if these vertices are not yet colored by $\gamma$, then color white can be assigned to $w_1,w_2$.
\vspace{-0.3cm}\item[(p)] If $G$ is $S_{2,2,2}$-free and contains a vertex $v$ such that the subgraph induced by $N(v)$ has three isolated vertices, then $v$ must be black.
\vspace{-0.3cm}\item[(q)] Suppose $G$ is $S_{2,2,2}$-free. If it contains $F_8$ as an induced subgraph, and if $d_G(x_1)=d_G(x_2)=3$, $d_G(w_i)=d_G(v_i)=2$, $i=1,2$, then $y_1$ and $y_2$ must be white.
\vspace{-0.3cm}\item[(r)] Suppose $G$ is butterfly-free and contains a vertex $v$ with four neighbors $w_1,w_2,w_3,w_4$ such that only two of them are adjacent, say $w_1$ and $w_2$. If
$G$ does not contain two vertices $u_1,u_2$ such that  $N(u_1)\cap \{v,w_1,w_2,w_3,w_4,u_2\}=\{w_3\}$ and $N(u_2)\cap \{v,w_1,w_2,w_3,w_4,u_1\}=\{w_4\}$,
then $v$ must be black.
\vspace{-0.3cm}\item[(s)] If $G$ contains $F_9$ as an induced subgraph and $x$ and $y$ are black, then $z$ must be black.
\vspace{-0.3cm}\item[(t)] If $G$ contains $F_{10}$ as an induced subgraph and $x$ is black, then $y$ must be white.
\end{itemize}
\end{lemma}
\vspace{-0.1cm}\begin{proof}
\begin{itemize}
\vspace{-0.3cm}\item[(a)] If $d_G(u)=0$, then $u$ must be white since $u$ cannot have a black neighbor. If $d_G(u)=1$, then the neighbor of $u$ cannot be white since otherwise $u$ would need to be black with no black neighbor, a contradiction.
\vspace{-0.3cm}\item[(b)] If a common neighbor $w$ of the two non-adjacent black vertices is black, then $w$ has two black neighbors, a contradiction.
\vspace{-0.3cm}\item[(c)] If the vertex shared by the two triangles is black, then, since the white vertices form a stable set, it must have at least two black neighbors, one in each triangle, a contradiction.
\vspace{-0.3cm}\item[(d)] If one of the two vertices shared by the two triangles is white, then the other one is black and must have at least two black neighbors, one in each triangle, a contradiction.
\vspace{-0.3cm}\item[(e)] It follows from (a) that $u$ is black in all $\gamma$-completions. Hence, at most one of its neighbors is black. We can therefore impose color white on $k-1$ of its neighbors of degree 1.
\vspace{-0.3cm}\item[(f)] Suppose to the contrary that $v_4$ is  white. Then $v_3$ is black and since $d_G(v_3)=2$, it follows that $v_2$ is black. But now $v_2$ has two black neighbors, a contradiction. 
\vspace{-0.3cm}\item[(g)] If $x$ is white, then $y$ is black and it follows from Lemma \ref{lem:samedifferent} (ii) that $z$ is white, which means that $y$ has no black neighbor, a contradiction.
\vspace{-0.3cm}\item[(h)] It follows from Lemma \ref{lem:samedifferent} (i) that $y$ and $z$ must get different colors, which means that $x$ is necessarily black.
\vspace{-0.3cm}\item[(i)] Suppose to the contrary that $x$ is white. Then both $w_1,w_2$ must be black, and  
it follows from Lemma \ref{lem:samedifferent} (ii) that both $u_1,u_2$ are white, a contradiction. Now,
\begin{itemize} 
\vspace{-0.2cm}\item if $d_G(x)=2$, then one of $w_1,w_2$ must be black, which implies that one of $u_1,u_2$ must
be white. Hence $y$ must be black;
\vspace{-0.2cm}\item if $y$ is black, then one of $u_1,u_2$ is white, which means that one of $w_1,w_2$ is the black neighbor of $x$. 
Hence all neighbors $z\neq w_1,w_2$ of $x$ are white;
\vspace{-0.2cm}\item if $d_G(u_1)=d_G(u_2)=3$ and $d_G(w_1)=d_G(w_2)=2$ and if these vertices are not yet colored by $\gamma$, then consider any $\gamma$-completion. If $w_1$ is black, then $u_1,w_2$ are white and $y,u_2$ are black. We can easily transform this $\gamma$-completion into another by coloring $u_1,w_2$ black and $u_2,w_1$ white.
\end{itemize}
\vspace{-0.3cm}\item[(j)] If follows from (g) that $x$ must be black. Suppose that $d_G(y)=d_G(z)=2$. If $y$ is white, then it follows from Lemma \ref{lem:samedifferent} (ii) that $u$ is black and $v$ is white, which means that $z$ is black. Hence, either $y$ or $z$ is black, which means that all other neighbors of $x$ must be white. 
\vspace{-0.3cm}\item[(k)] Clearly, exactly one of $x_1,x_2,x_3$ must be white, say $x_1$. Then $y_1$ and $v$ are black (since $d_G(y_1)=2$), which means that $u$ must be white.
\vspace{-0.3cm}\item[($\ell$)] Suppose to the contrary that $v$ is black. Lemma 2.1 (i) implies that the neighbors of $v$ are white, that is $v$ has no black neighbor, a contradiction.
\vspace{-0.3cm}\item[(m)] Suppose to the contrary that $x$ is black. It then follows from  Lemma 2.1 (ii) that $y$ is white. Hence, $z$ must be the black neighbor of $v$, and $w_1,w_2$ must be white. But then all three vertices of the triangle induced by $u_1, u_2, x$ are black, which is impossible.
\vspace{-0.3cm}\item[(n)] Consider any $\gamma$-completion. It follows from (g) that $u_1$ and $u_2$ are black.  By Lemma 2.1 (ii), $v_1$ has the same color as $v_2$ and $w_1$ has the same color as $w_2$. Note that at least one of these pairs of vertices must be white, otherwise $u_1$ (and $u_2$) would have two black neighbors. Assume that $w_1$ and $w_2$ are black, then $v_1, v_2$ are white, $x,y$ are black and $z$ is white. Since $d_G(y)=d_G(z)=4$, $d_G(v_i)=d_G(w_i)=2$, $i=1,2$, we can easily transform this $\gamma$-completion into another 
by recoloring $v_1,v_2,z$ black and $w_1,w_2,y$ white.
\vspace{-0.31cm}\item[(p)] Consider three isolated vertices $x,y,z$ in $N(v)$, and suppose $v$ is white. It follows that $x,y,z$ are black and hence each has a neighbor not in $N(v)$ which must also be black. 
Since these neighbors must be distinct and non-adjacent, we obtain an induced $S_{2,2,2}$, a contradiction.
\vspace{-0.31cm}\item[(q)] It follows from (p) that $x_1$ and $x_2$ must be black. Hence, exactly one of $w_1,w_2$ and exactly one of $u_1,u_2$ must be black, which means that $y_1$ and $y_2$ must be white.
\vspace{-0.31cm}\item[(r)] Suppose to the contrary that $v$ is white. Then $w_1,w_2,w_3,w_4$ are black. Let $u_1$ the the black neighbor of $w_3$ and $u_2$ be the black neighbor of $w_4$ in a $\gamma$-completion. Since $G$ is butterfly-free and every black vertex has exactly one black neighbor, we have $N(u_1)\cap \{v,w_1,w_2,w_3,w_4,u_2\}=\{w_3\}$ and $N(u_2)\cap \{v,w_1,w_2,w_3,w_4,u_1\}=\{w_4\}$, a contradiction.
\vspace{-0.31cm}\item[(s)] If $x$ and $y$ are black, then one of $w_1,w_2$ must be black. Hence one of $u_1,u_2$ must be white, which implies that $z$ must be black.
\vspace{-0.3cm}\item[(t)] If $y$ is black, then $u_1,u_2$ must be white, $v_1,v_2$ must be black and $w_1,w_2$ are then two white adjacent vertices, a contradiction.
\end{itemize}
\vspace{-0.8cm}\end{proof}

\vspace{0.4cm}In addition to the above forcing rules, we will also use the following ones which are clearly valid : 
\begin{itemize}
\vspace{-0.3cm}\item[(i)] If a vertex $v$ is white, then all its neighbors must be black.
\vspace{-0.3cm}\item[(ii)] If two adjacent vertices are black, then all their neighbors must be white.
\vspace{-0.3cm}\item[(iii)] If a vertex $u$ is black, and all its neighbors, except $v$, are white, then $v$ must be black.
\end{itemize}

If one of the rules (i), (ii), (iii), or one of the rules described in Lemmas \ref{lem:samedifferent} and \ref{lem:simpleforcing} imposes color black (resp. white) on a vertex that is already forced to be white (resp. black), we conclude that the considered graph does not admit a feasible complete coloring. 
Applying these rules repeatedly, as often as possible, on a given graph $H$, we thus either get a proof that $H$ does not 
admit a feasible complete coloring, or we obtain a feasible partial coloring of $H$. 
In the latter case, we can apply a cleaning to obtain a graph $G$ with a feasible partial coloring $\gamma$ so that $W_{\gamma}=\emptyset$ 
and the distance between any two vertices of $B_{\gamma}$ is at least 3. Indeed, $B_{\gamma}$ is a stable set since adjacent black 
vertices are removed by a cleaning, and two vertices $u,v$ in $B_{\gamma}$ cannot have a common neighbor $w$ since 
Lemma \ref{lem:simpleforcing} (b) would impose color white on $w$, and $w$ would therefore be removed by a cleaning. 
This justifies the following definition.

\begin{definition}
Let $\gamma$ be a feasible partial coloring of a graph $G$ such that $W_{\gamma}=\emptyset$ and the distance between any two vertices of $B_{\gamma}$ is at least 3. The pair $(G,\gamma)$ is called {\sc clean} if none of the forcing rules defined above can color additional vertices.
\end{definition}

\begin{rmks}\label{remarks}
\end{rmks}\vspace{-0.3cm}
{\em\begin{itemize}
\vspace{-0.3cm}\item[(a)] Rules (c) and (d) of Lemma \ref{lem:simpleforcing} show that if $(G,\gamma)$ is clean, then $G$ does not contain any induced 
diamond and any induced butterfly.
\vspace{-0.3cm}\item[(b)] It is easy to see that a graph containing a $K_4$ cannot admit a feasible complete coloring.
\end{itemize}}

Hence, it follows from the remarks above that we may suppose that all considered graphs have no induced diamond, no induced butterfly and are $K_4$-free. 
Note that if $(G,\gamma)$ is clean and was obtained from an $S_{2,2,2}$-free graph $H$ by applying the above-mentioned forcing rules followed by a cleaning, 
then $G$ is an induced subgraph of $H$ and is therefore also $S_{2,2,2}$-free. The following lemma gives additional properties of clean pairs.

\begin{lemma}\label{lem:cleanPairProperties}
Let $(G,\gamma)$ be a clean pair. If $\gamma$ is completable, then the following claims hold.
\begin{enumerate}
\vspace{-0.3cm}\item[(a)] Each vertex of $G$ belongs to at most one triangle.
\vspace{-0.3cm}\item[(b)] If $G$ contains $F_3$ as an induced subgraph, then $x\in B_{\gamma}$, the degree of any neighbor of $x$ is at most two, and $x$ does not belong to any triangle. 
\vspace{-0.3cm}\item[(c)] If $G$ contains $F_2$ as an induced subgraph, then $x$ is black and $x$ has no other neighbors.
\vspace{-0.3cm}\item[(d)] Let $T_1$ and $T_2$ be two vertex-disjoint triangles in $G$. 
Then there are at most two edges between $T_1$ and $T_2$. Moreover, if there are exactly two edges between the triangles, then these two edges are not adjacent.
\end{enumerate}
\end{lemma}

\begin{proof}
\begin{enumerate}
\vspace{-0.3cm}\item[(a)] This is a direct consequence of the Remarks \ref{remarks} and the fact that $(G,\gamma)$ is clean.
\vspace{-0.3cm}\item[(b)] Lemma \ref{lem:simpleforcing} (i) implies that $x\in B_{\gamma}$, which means that no neighbor of $x$ belongs to $B_{\gamma}\cup W_{\gamma}$.  
Assume $w_1$ has a neighbor $z$ different from $u_1$ and $x$. If $z$ is adjacent to $x$, then (a) and Lemma \ref{lem:samedifferent} (ii) imply that $w_2\in W_{\gamma}$, a contradiction. Hence, no neighbor of $w_1$ is adjacent to $x$. Now, if $x$ belongs to a triangle then Lemma \ref{lem:samedifferent} (ii) implies that $w_1\in W_{\gamma}$, a contradiction. Finally, let $v\neq w_1,w_2$ be a neighbor of $x$. If $d_G(v)\geq 3$, then  Lemma \ref{lem:samedifferent} (ii) and  Lemma \ref{lem:simpleforcing} (p) imply that $v\in B_{\gamma}\cup W_{\gamma}$, a contradiction.
\vspace{-0.3cm}\item[(c)] Lemma \ref{lem:simpleforcing} (h) implies that $x\in B_{\gamma}$. Assume that $x$ has a neighbor $s$ different from 
$y$ and $z$. Then (a) and Lemma \ref{lem:samedifferent} (ii) imply that $s\in W_{\gamma}$, a contradiction.
\vspace{-0.3cm}\item[(d)]By (a), all edges between $T_1$ and $T_2$ are pairwise nonadjacent. 
Hence there are at most three edges between the triangles. By Lemma \ref{lem:samedifferent} (ii) each of the edges connects 
vertices of different colors. Therefore there are at most two edges between $T_1$ and $T_2$, otherwise one of the triangles would have 
two white vertices, a contradiction.
\end{enumerate}
\vspace{-0.8cm}\end{proof}


\section{Graph reductions}
\label{sec:REDUCTION}


\begin{definition}
Let $(G,\gamma)$ be a clean pair, $G'$ an induced subgraph of $G$ and $\delta$ the restriction of $\gamma$ to $G'$. 
The replacement of $(G,\gamma)$ by $(G',\delta)$ is a {\sc valid reduction} if either both $\gamma$ and $\delta$ are completable, or none of them is.
\end{definition}

In this section, we will present eight valid reductions  $\rho_1,\dots,\rho_8$. Assume $G=(V,E)$ contains one of the graphs $H_i$ ($i=1,\dots,8$) of Figure \ref{figH1a8} as induced subgraph, where each of the dashed edges can be replaced by a true edge or a non-edge. Let $i\in \{1,\ldots,8\}$. Let $U_i$ be the set of grey vertices in $H_i$, and assume that no vertex in $U_i$ has other neighbors in $G$ than those in $H_i$. Finally, let $\delta_i$ be the restriction of $\gamma$ to $G[V\setminus U_i]$. Reduction $\rho_i$ consists in replacing $(G,\gamma)$ by $(G[V\setminus U_i],\delta_i)$. Note that if $G$ is $S_{2,2,2}$-free, then the graph obtained by applying reduction $\rho_i$ is also $S_{2,2,2}$-free since it is an induced subgraph of $G$.

\begin{figure}[!ht]
\begin{center}
\scalebox{0.163}{\includegraphics{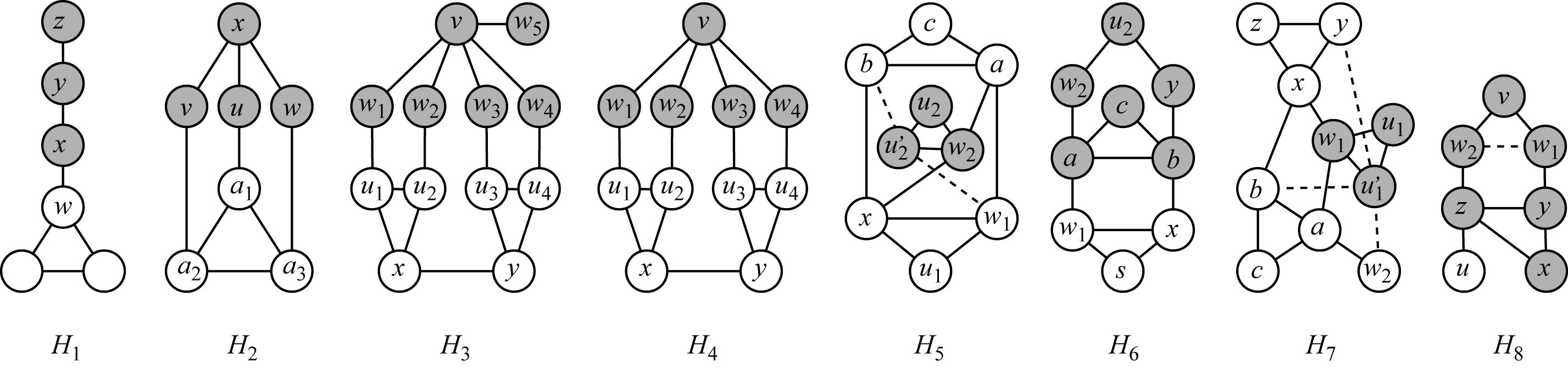}}\caption{Eight reductions.}
\label{figH1a8}
\end{center}
\end{figure}

\vspace{-0.5cm}\begin{lemma}\label{lem:8reductions}
Reductions $\rho_1,\dots,\rho_8$ are valid.
\end{lemma}

\vspace{-0.3cm}\begin{proof}  
Let $i\in \{1,\ldots,8\}$. First observe that ($G[V\setminus U_i],\delta_i$) is clean since ($G,\gamma$) is clean. Now, let $S_i$ be the set of non-grey vertices in $H_i$, which means that all neighbors of the vertices in $U_i$ belong to $S_i$. Let ${\bar \gamma}$ be a $\gamma$-completion and ${\bar \delta_i}$ the restriction of ${\bar \gamma}$ to $G[V\setminus U_i]$.
Consider any two adjacent vertices $v_1,v_2$ in $H_i$ such that $v_1\in U_i$ and $v_2\in S_i$. 
Note that if $v_1$ belongs to a triangle in $G[U_i]$, then $v_2$ is nonadjacent to the other two vertices of that 
triangle, while if $v_2$ belongs to a triangle in $G[S_i]$, then $v_1$ is nonadjacent to the other two vertices of that 
triangle. It then follows from Lemma 2.1 (ii) that ${\bar \gamma}(v_1)\neq {\bar \gamma}(v_2)$, which implies that 
${\bar \delta_i}$ is feasible and thus $\delta_i$ is completable.\\



Let now ${\bar \delta_i}$ be a $\delta_i$-completion. We show how to extend ${\bar \delta_i}$ to a $\gamma$-completion.

\begin{itemize}
\vspace{-0.3cm}\item $i=1$. Lemma \ref{lem:simpleforcing} (a) implies that $y\in B_{\gamma}$. Hence, none of $x,z$ belongs to $B_{\gamma}$. 
If ${\bar \delta_1}(w)=$ black, we obtain a $\gamma$-completion by assigning color black to $y,z$, and color white to $x$. 
If ${\bar \delta_1}(w)=$ white, a $\gamma$-completion is obtained by assigning color black to $x,y$, and color white to $z$.
\vspace{-0.3cm}\item $i=2$. Lemma \ref{lem:simpleforcing} (p) implies that $x\in B_{\gamma}$. Hence, none of $u,v,w$ belongs to $B_{\gamma}$.
Without loss of generality, we may assume that ${\bar \delta_2}(a_1)=$ white, and we can obtain a $\gamma$-completion 
by assigning color black to $x,u$, and color white to $v,w$.
\vspace{-0.3cm}\item $i=3$ or $4$. Lemma \ref{lem:simpleforcing} (p) implies that $v\in B_{\gamma}$. 
Hence, none of the $w_i$'s belongs to $B_{\gamma}$. Lemma \ref{lem:samedifferent} (ii) 
implies that exactly one of $x,y$ is black in ${\bar \delta_i}$, say $x$. Then exactly one of $u_1,u_2$ is white in ${\bar \delta_i}$, say $u_1$. Hence $u_2,u_3,u_4$ 
are black and $y$ is white in ${\bar \delta_i}$. We can then obtain a $\gamma$-completion by assigning color black to $w_1,v$, and color white to $w_2,w_3,w_4$, and to $w_5$ if $i=3$.
\vspace{-0.3cm}\item $i=5$. Lemma \ref{lem:samedifferent} (ii) implies that $w_2 \notin B_{\gamma}$ (else $x$ and $a$ would belong to $W_{\gamma}$), 
and Lemma \ref{lem:samedifferent} (i) implies ${\bar \delta_5}(x) = {\bar \delta_5}(a) \neq {\bar \delta_5}(b) = {\bar \delta_5}(w_1)$. 
If $u_2' \in B_{\gamma}$ then $u_2 \notin B_{\gamma}$ and $b$ and $w_1$ are not neighbors of $v$ (else they would belong to $W_{\gamma}$). 
We can then obtain a $\gamma$-completion by assigning color ${\bar \delta_5}(x)$ to $u_2$, color ${\bar \delta_5}(w_1)$ to $w_2$ and color black to $u_2'$. 
If $u _2'\notin B_{\gamma}$, we obtain a $\gamma$-completion by assigning color ${\bar \delta_5}(x)$ to $u_2'$, color ${\bar \delta_5}(w_1)$ 
to $w_2$ and color black to $u_2$.
\vspace{-0.3cm}\item $i=6$. Lemma \ref{lem:samedifferent} (ii) implies that $\{a,b,y,w_2\}\cap B_{\gamma}=\emptyset$, (otherwise two of them would belong to $W_{\gamma}$).
Also, Lemma \ref{lem:simpleforcing} (h) implies that $s \in B_{\gamma}$. Hence ${\bar \delta_6}(x)\neq {\bar \delta_6}(w_1)$, 
and a $\gamma$-completion is obtained by assigning color ${\bar \delta_6}(x)$ to $a,y$, color ${\bar \delta_6}(w_1)$ to $b, w_2$ and 
color black to $c,u_2$. 
\vspace{-0.3cm}\item $i=7$. Lemma \ref{lem:simpleforcing} (h) and Lemma \ref{lem:samedifferent} (i) imply that $c \in B_{\gamma}$ and ${\bar \delta_7}(b) = {\bar \delta_7}(w_2) \neq {\bar \delta_7}(a)={\bar \delta_7}(x)$. 
If $b,y$ and $w_2$ are not adjacent to $u_1'$, then either $u_1\in B_{\gamma}$ and a $\gamma$-completion is obtained
by assigning color black to $u_1$, color ${\bar \delta_7}(b)$ to $w_1$ and color ${\bar \delta_7}(a)$ to $u_1'$, 
or $u_1\notin B_{\gamma}$ and a $\gamma$-completion is obtained by assigning color black to $u_1'$, color ${\bar \delta_7}(b)$ to $w_1$ and color ${\bar \delta_7}(a)$ to $u_1$.
So assume at least one of $b,y,w_2$ is adjacent to $u_1'$. Then Lemma \ref{lem:simpleforcing} (h) implies that $u_1\in B_{\gamma}$.
Note that if $u_1'$ is adjacent to $y$ then Lemma \ref{lem:simpleforcing} (h) and Lemma \ref{lem:samedifferent} (ii) imply that$z \in B_{\gamma}$ and 
${\bar \delta_7}(y)={\bar \delta_7}(b)\neq {\bar \delta_7}(a)$. Hence, a $\gamma$-completion is obtained by assigning color black to $u_1$, color ${\bar \delta_7}(b)$ to $w_1$ 
	and color ${\bar \delta_7}(a)$ to $u_1'$. 
\vspace{-0.3cm}\item $i=8$. Lemma \ref{lem:samedifferent} (ii) implies that $\{w_1,w_2,y,z\}\cap B_{\gamma}=\emptyset$ (else $W_{\gamma}\neq \emptyset$). 
If ${\bar \delta_8}(u)=$black, then a $\gamma$-completion is obtained by assigning color white to 
	$z,w_1$ and color black to $x,y,w_2,v$. If ${\bar \delta_8}(u)=$white, then a $\gamma$-completion is obtained 
	by assigning color white to $y,w_2$ and 
	color black to $x,z,w_1,v$.  
\end{itemize}
\vspace{-0.5cm}\end{proof}


\section{Graph transformations}
\label{sec:TRANSFORMATION}


Let $G=(V,E)$ be a graph, and let $\gamma$ be a feasible partial coloring of $G$. Let $G'$ be a graph obtained from $G$ by removing a subset $X$ of its vertices, adding a subset $Y$ of new vertices, adding or/and removing some edges in $G[V\setminus X]$, and finally 
adding some edges linking pairs of vertices in $Y$ as well as some edges linking some vertices in $Y$ with some vertices in $V\setminus X$. Such an operation is called a \textit{graph transformation}. The restriction $\delta$ of $\gamma$ to $G'$ is defined as the partial coloring of $G'$ obtained by setting $\delta(v)=\gamma(v)$ for all vertices in $V\setminus X$, and 
by leaving all vertices in $Y$ uncolored.

\begin{figure}[!ht]
\begin{center}
\scalebox{0.165}{\includegraphics{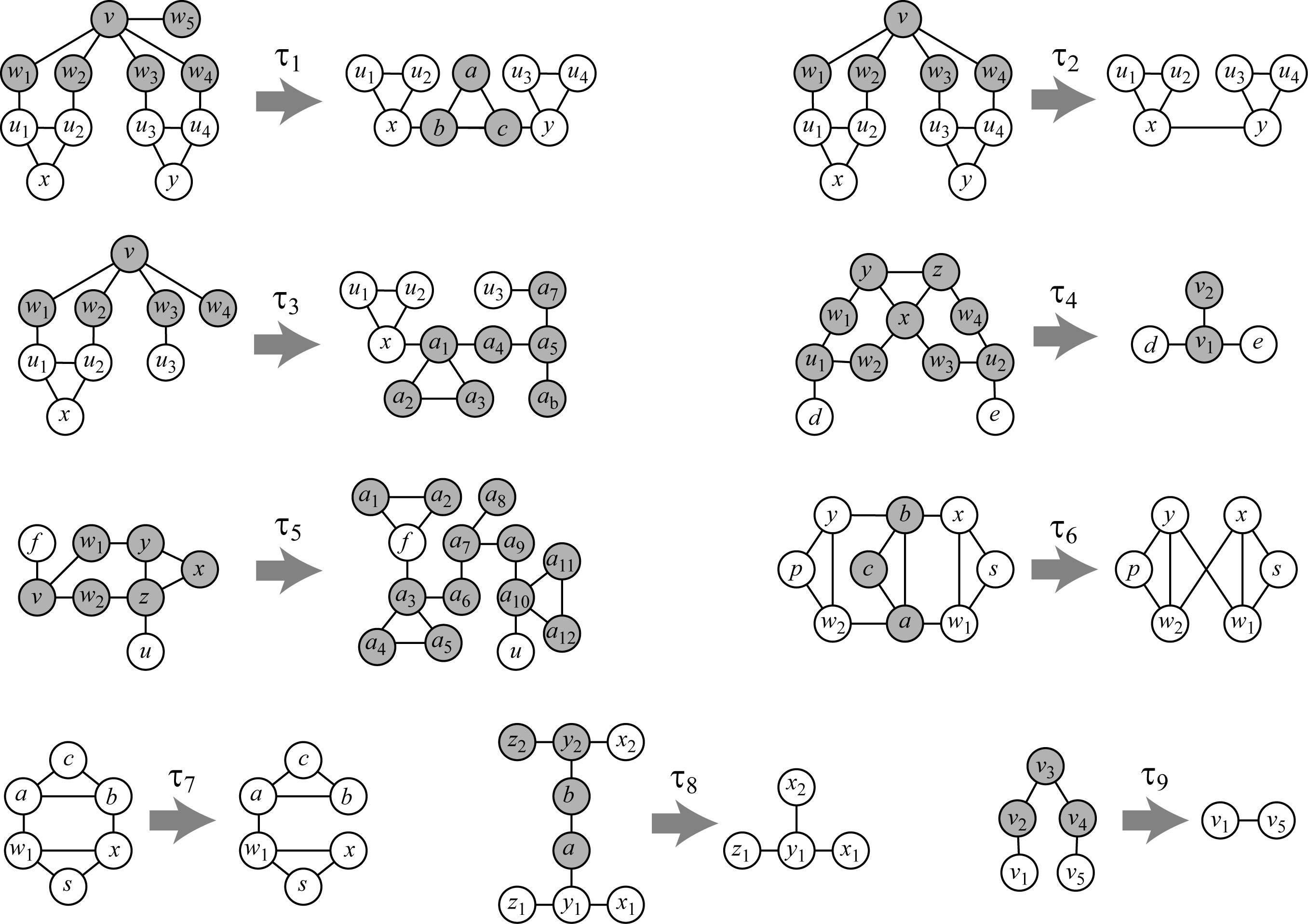}}\caption{Nine graph transformations.}
\label{fig:Transformations}
\end{center}
\end{figure}

\begin{definition}
Let $(G,\gamma)$ be a clean pair where $G$ is $S_{2,2,2}$-free. Let $G'$ be a graph obtained from $G$ by applying some graph transformation, and let $\delta$ be the restriction of $\gamma$ to $G'$. The replacement of $(G,\gamma)$ by $(G',\delta)$ is a {\sc valid transformation} if $G'$ is $S_{2,2,2}$-free and either both $\gamma$ and $\delta$ are completable, or none of them is.
\end{definition}

Nine graph transformations $\tau_{1},\dots,\tau_{9}$ are represented in Figure \ref{fig:Transformations}. For every transformation, we show on the left an induced subgraph of $G$ while modifications made on $G$ to obtain $G'$ appear on the right. The set $X$ of removed vertices from $G$ and the set $Y$ of added vertices to $G'$ are shown in grey.
No vertex in $X$ (resp. $Y$) has other neighbors in $G$ (resp. $G'$) than those shown in Figure \ref{fig:Transformations}. In the following lemmas, we assume that $(G,\gamma)$ is a clean pair, that $G$ is $S_{2,2,2}$-free, and that $\delta$ is the restriction of $\gamma$ to $G'$. When constructing a $\gamma$-completion ${\bar \gamma}$ from a $\delta$-completion ${\bar \delta}$, or the opposite, we will always assume ${\bar \gamma}(v)={\bar \delta}(v)$ for all $v\in V\setminus X$, unless otherwise specified.

\vspace{0.3cm}\begin{lemma}
\label{lem:T1}
Transformation $\tau_{1}$ is valid.
\end{lemma}

\vspace{-0.3cm}\begin{proof}
Lemma \ref{lem:simpleforcing} (p) implies that $v\in B_{\gamma}$, which means that $u_i,w_i\not\in B_{\gamma}$ for $i=1,\ldots,4$, $j=1\ldots,5$. 
Suppose by contradiction that $G'$ contains an induced $S_{2,2,2}$, and denote this $S_{2,2,2}$ by $H$. Since $G$ is $S_{2,2,2}$-free, we may assume without loss of generality that 
$x$ is the vertex of degree 3 in $H$, and either $u_1$ or $u_2$, say $u_1$ is a neighbor of $x$ in $H$. 
In other words, $G'$ contains three vertices $z_1,z_2,z_3$ such that $x,b,c,u_1,z_1,z_2,z_3$ induce $H$ in $G'$, with $z_1\in N(u_1)$, 
and $\{z_3,x\}\subseteq N(z_2)$. Note that Lemma \ref{lem:cleanPairProperties} (a) implies that $u_2$ is not adjacent to $z_1,z_2$. 
Since $u_1\notin B_{\gamma}$ (otherwise $w_1\in W_{\gamma}$), Lemma \ref{lem:simpleforcing} (a) implies that $d_G(z_1)>1$. 
So let $z_4\neq u_1$ be another neighbor to $z_1$. Lemma \ref{lem:cleanPairProperties} (a) implies that $z_4$ is not adjacent 
to $u_1$, and Lemma \ref{lem:cleanPairProperties} (a) and (c) imply that $z_4$ is not adjacent 
to $u_2$ and $x$. Finally, Lemma \ref{lem:cleanPairProperties} (b) implies that $z_4$ is not adjacent to $z_2$, which means that
$v,w_1,u_1,z_1,z_4,x,z_2$ induce an $S_{2,2,2}$ in $G$, a contradiction.

\vspace{0.1cm}Let now ${\bar \gamma}$ be a $\gamma$-completion. If $w_5$ is black, then $w_1,w_2,w_3,w_4,x,y$ are white, $u_1,u_2,u_3,u_4$ are black, and we obtain a $\delta$-completion 
by assigning color black to $b,c$ and color white to $a$. 
If $w_5$ is white, then one of $w_1,w_2,w_3,w_4$ is black, say $w_1$, which means that $u_2,u_3,u_4,x$ are black, $u_1,y$ are white, and we obtain a $\delta$-completion 
by assigning color black to $a,c$ and color white to $b$.

\vspace{0.1cm}Finally, let ${\bar \delta}$ be a $\delta$-completion. 
Note that at least one of $x,y$ is white. Indeed, if $x$ is black, then Lemma \ref{lem:samedifferent} (ii) 
implies that $b$ is white, which means that $a$ and $c$ are black and $y$ is white. 
Hence, at most one of $u_1,u_2,u_3,u_4$ is white. If none of them is white, we obtain a $\gamma$-completion by assigning 
color black to $w_5,v$ and color white to $w_1,w_2,w_3,w_4$. If one of $u_1,u_2,u_3,u_4$ is white, say $u_1$, we obtain a $\gamma$-completion 
by assigning color black to $w_1,v$ and color white to $w_2,w_3,w_4,w_5$.
\vspace{-0.5cm} \end{proof}

\vspace{0.3cm}\begin{lemma}\label{lem:T2}
	Transformation $\tau_{2}$  is valid.
\end{lemma}
\vspace{-0.3cm}\begin{proof}
The proof that $G'$ is $S_{2,2,2}$-free is the same as the one in Lemma \ref{lem:T1}, where $b,c$ are replaced by $y$ and a neighbor $z\neq x$ of $y$.
Hence, we only show that $\gamma$ is completable if and only if $\delta$ is completable.
Lemma \ref{lem:simpleforcing} (p) implies that $v\in B_{\gamma}$, which means that $u_i,w_i\not\in B_{\gamma}$ for $i=1,\ldots,4$.

\vspace{0.1cm}Let ${\bar \gamma}$ be a $\gamma$-completion. Exactly one of $w_1,w_2,w_3,w_4$ is black, which implies that exactly one of $u_1,u_2,u_3,u_4$ and one of $x,y$ is white.
Hence, a $\delta$-completion is obtained by coloring all vertices of $G'$ as in $G$.
	
\vspace{0.1cm}Let now ${\bar \delta}$ be a $\delta$-completion. Lemma \ref{lem:samedifferent} (ii) implies that $x$ and $y$
		have different colors. Hence, exactly one of $u_1,u_2,u_3,u_4$ is
		white, say $u_1$, and we obtain a $\gamma$-completion 
		by assigning color black to $w_1,v$ and color white to $w_2,w_3,w_4$.
\vspace{-0.5cm}\end{proof}

\vspace{0.3cm}\begin{lemma}\label{lem:T3}
	Transformation $\tau_{3}$  is valid.
\end{lemma}
\vspace{-0.3cm}\begin{proof}
The proof that $G'$ is $S_{2,2,2}$-free is the same as the one in Lemma \ref{lem:T1}, where $b,c$ are replaced by $a_1,a_2$.
Hence, we only show that $\gamma$ is completable if and only if $\delta$ is completable.
Lemma \ref{lem:simpleforcing} (a) implies that $v\in B_{\gamma}$, which means that $u_i,w_j\not\in B_{\gamma}$ for $i=1,2,3$, $j=1,\ldots,4$. 

\vspace{0.1cm}Let ${\bar \gamma}$ be a $\gamma$-completion. Exactly one of $w_1,w_2,w_3,w_4$ is black. If $w_4$ is black, then $w_1,w_2,w_3,x$
		are white, $u_1,u_2,u_3$ are black, and we obtain a $\delta$-completion 
		by assigning color black to $a_1,a_2,a_5,a_6$ and color white to $a_3,a_4,a_7$.
		If $w_3$ is black, then $w_1,w_2,w_4,x,u_3$ are white, $u_1,u_2$ are
		black, and we obtain a $\delta$-completion 
		by assigning color black to $a_1,a_2,a_5,a_7$ and color white to $a_3,a_4,a_6$.
		Finally, if one of $w_1,w_2$ is black, say $w_1$,
		then $u_1,w_2,w_3,w_4$ are white, $u_2,u_3,x$ are black, and we obtain
		a $\delta$-completion by assigning color black to $a_2,a_3,a_4,a_5$ and color white to $a_1,a_6,a_7$.
	
\vspace{0.1cm}Let now ${\bar \delta}$ be a $\delta$-completion. Note that $u_3$ is black whenever $x$ 
		is black. Indeed, if $x$ is black, then it follows from Lemma \ref{lem:samedifferent} (ii) that 
		$a_1$ is white and $a_4$ is black, which means that $a_5$ is black, 
		$a_7$ is white, and $u_3$ is black. Hence, at most one of $u_1,u_2,u_3$ is 
		white. If none of them is white, we obtain a $\gamma$-completion by assigning color black to $w_4,v$ and color white 
		to $w_1,w_2,w_3$. If $u_3$ is white, we obtain a
		$\gamma$-completion by assigning color black to 
		$w_3,v$ and color white to $w_1,w_2,w_4$. We proceed in a similar way if $u_1$ or $u_2$ is white.
\vspace{-0.5cm}\end{proof}

\vspace{0.3cm}\begin{lemma}\label{lem:T4}
Transformation $\tau_{4}$  is valid.
\end{lemma}
\vspace{-0.3cm}\begin{proof}
Suppose $G'$ contains an induced $S_{2,2,2}$, and denote this $S_{2,2,2}$ by $H$. Since $G$ is $S_{2,2,2}$-free, we may assume without loss of generality that
$d$ and $v_1$ are vertices in $H$. We then get an $S_{2,2,2}$ in $G$ as follows: (i) if $d_H(d)=3$, we replace $v_1$ and its neighbor in $H$ by $u_1,w_1$; (ii) if $d_H(d)=2$, we replace
$v_1$ by $u_1$; (iii) if $d_H(d)=1$, we replace $d,v_1$ by $w_3,u_2$, a contradiction. Thus $G'$ is $S_{2,2,2}$-free.
	
	\vspace{0.1cm}Let ${\bar \gamma}$ be a $\gamma$-completion. Lemma \ref{lem:simpleforcing} (p) implies that $u_1,u_2\in B_{\gamma}$, 
	which means that none of $w_1,w_2,w_3,w_4,x,y,z,d,e$ belongs to $B_{\gamma}$.
	At most one of $d,e$ is black, else $x,y,z $ would necessarily be black as well, a contradiction.
	If $d$ and $e$ are white, a 
		$\delta$-completion is obtained by assigning color black to $v_1,v_2$.
	If $d$ and $e$ have different colors, a $\delta$-completion is obtained by 
		assigning color black to $v_1$ and color white to $v_2$. 
	
\vspace{0.1cm}Let now ${\bar \delta}$ be a $\delta$-completion. Lemma \ref{lem:simpleforcing} (a) implies that $v_1$ is black and 
		hence at most one of $d$ and $e$ is black. If $d$ and $e$ are white,
		we obtain a $\gamma$-completion by assigning 
		color black to $u_1,u_2,w_2,w_3,y,z$ and color white to $w_1,x,w_4$.
		If $d$ and $e$ have different colors, say $d$ is black and $e$ is white, 
		we obtain a $\gamma$-completion by assigning 
		color black to $u_1,u_2,x,y,w_4$ and color white to $w_1,w_2,w_3,z$.
\vspace{-0.5cm}\end{proof}

\vspace{0.3cm}\begin{lemma}\label{lem:T5}
		Transformation $\tau_{5}$  is valid.
\end{lemma}

\vspace{-0.3cm}\begin{proof}
Suppose $G'$ contains an induced $S_{2,2,2}$, and denote this $S_{2,2,2}$ by $H$. 
Since $G$ is $S_{2,2,2}$-free and $d_G(f)\leq 2$ (by Lemma \ref{lem:cleanPairProperties} (b)), 
$u$ is the vertex of degree 3 in $H$, 
	$a_{10}$ is a neighbor of $u$ in $H$, and without loss of generality, $a_9$ is the neighbor of $a_{10}$ in $H$. In other word, $G'$ contains
	four vertices $s_1,s_2,p_1,p_2$ such that $u,a_{10},a_9,s_1,s_2,p_1,p_2$ induce an
	$S_{2,2,2}$ in $G'$. But then $u,z,x,s_1,s_2,p_1,p_2$
 induce an $S_{2,2,2}$ in $G$, a contradiction.

\vspace{0.1cm}Let ${\bar \gamma}$ be a $\gamma$-completion. Lemmas \ref{lem:samedifferent} (ii) and \ref{lem:simpleforcing} (p) imply that $w_2,u$ have the same
color, and $v\in B_{\gamma}$.
Hence, at most one of $f,u$ is black, else the black vertex $v$ would have two black neighbors $f$ and $w_2$.
If $f$ is black and $u$ is white , we obtain a $\delta$-completion by 
assigning color black to $a_1,a_4,a_5,a_6,a_7,a_{10},a_{12}$ and color white to $a_2,a_3,a_8,a_9,a_{11}$.	 
If $f$ is white and $u$ is black, we obtain a $\delta$-completion by 
assigning color black to $a_1,a_2,a_3,a_4,a_7,a_9,a_{11},a_{12}$ and color white to $a_5,a_6,a_8,a_{10}$.
If both $f$ and $u$ are white, we obtain a $\delta$-completion by 
assigning color black to $a_1,a_2,a_3,a_4,a_7,a_8,a_{10},a_{11}$ and color white to $a_5,a_6,a_9,a_{12}$.

\vspace{0.1cm}Let now ${\bar \delta}$ be a $\delta$-completion. Lemma \ref{lem:samedifferent} (ii) implies that $\{w_1,w_2,y,z\}\cap B_{\gamma}=\emptyset$.
In $G'$, at most one of $f,u$ can be black. Indeed, if $f$ and $u$ are black, then Lemma \ref{lem:samedifferent} (ii) implies that $a_6$ and $a_9$
		are black as well. But this is impossible, since $a_7$ is black by Lemma \ref{lem:simpleforcing} (p).
Now if $f$ is black and $u$ is white, then Lemma \ref{lem:simpleforcing} (i) implies $x\notin B_{\gamma}$, and we obtain a $\gamma$-completion
		by assigning color black to $v,y,z$ and color white to $w_1,w_2,x$. 
If $f$ is white and $u$ is black, a $\gamma$-completion is obtained by
assigning color black to $v,w_2,y,x$ and color white to $w_1,z$. 
Finally, if $f$ and $u$ are both white, a $\gamma$-completion is obtained by
assigning color black to $v,w_1,z,x$ and color white to $w_2,y$.  
\vspace{-0.5cm}\end{proof}

\vspace{0.3cm}\begin{lemma}\label{lem:T6}
		Transformation $\tau_{6}$  is valid.
\end{lemma}

\vspace{-0.3cm}\begin{proof}
	Suppose $G'$ contains an induced $S_{2,2,2}$, and denote this $S_{2,2,2}$ by $H$. Since $G$ is $S_{2,2,2}$-free,
	$H$ contains at least one of the
	new edges $yw_1$ and $xw_2$. In fact, $H$ contains exactly one of these edges,
	because $H$ is $C_4$-free. Without loss of generality, assume that $H$ contains
	$xw_2$. If $d_{H}(x) = 2$ and $d_{H}(w_2) = 1$ (resp. $d_{H}(w_2) = 2$ and $d_{H}(x) = 1$), then by replacing $w_2$ (resp. $x$) in $H$ by $b$ (res. $a$) we 
	obtain an induced $S_{2,2,2}$ in $G$, a contradiction. If $d_{H}(x) = 3$ and $d_{H}(w_2) = 2$ (resp. $d_{H}(w_2) = 3$ and $d_{H}(x) = 2$),
	then by replacing $w_2$ (resp. $x$) and the neighbor of $w_2$ (resp. $x$) of degree one in $H$ by
	$b$ and $c$ (resp. $a$ and $c$), we obtain an induced $S_{2,2,2}$ in $G$, a contradiction.
	
	\vspace{0.1cm}Let now ${\bar \gamma}$ be a $\gamma$-completion. 
	It follows from Lemma \ref{lem:samedifferent} (i) that ${\bar \gamma}(y)={\bar \gamma}(x)\neq {\bar \gamma}(w_1)={\bar \gamma}(w_2)$. Hence
	a $\delta$-completion can be obtained by coloring every vertex of $G'$ as in $G$.
	
	\vspace{0.1cm}Finally, let ${\bar \delta}$ be a $\delta$-completion. Lemma \ref{lem:samedifferent} (ii) implies that $\{a,b\}\cap B_{\gamma}=\emptyset$, and 
	Lemma \ref{lem:samedifferent} (i) implies ${\bar \delta}(x) = {\bar \delta}(y) \neq {\bar \delta}(w_1) = {\bar \delta}(w_2)$. We therefore obtain
	a $\gamma$-completion by assigning color ${\bar \delta}(x)$ to $a$, color ${\bar \delta}(w_1)$ to $b$, and color black to $c$.
\end{proof}

\vspace{0.3cm}\begin{lemma}\label{lem:T7}
		If $d_G(b)=3$ then transformation $\tau_{7}$ is valid.
\end{lemma}

\vspace{-0.3cm}\begin{proof}
	First notice that Lemma \ref{lem:cleanPairProperties} (c) implies that $d_G(c)=d_G(s)=2$. Suppose $G'$ contains an induced $S_{2,2,2}$, and denote 
	this $S_{2,2,2}$ by $H$. Since $G$ is $S_{2,2,2}$-free, $H$ must contain both $b$ and $x$. Since $b$ and $c$ cannot have degree $2$ or $3$ in $H$, we have $d_{H}(b)=1$ and $d_{H}(a)=2$. 
	But then by replacing $b$ in $H$ with $c$, we obtain an induced 
	$S_{2,2,2}$ in $G$, a contradiction.
	
	In order to show that $\gamma$ is completable if and only if $\delta$ is completable, 
	it is sufficient to prove that all $\gamma$-completions and $\delta$-completions assign different
	colors to $b$ and $x$. For a $\gamma$-completion this is 
	guaranteed by 
	Lemma \ref{lem:samedifferent} (ii). Now let ${\bar \delta}$ be a $\delta$-completion. 
	Lemma \ref{lem:simpleforcing} (h) implies that $c,s \in B_{\gamma}$. Hence, 
	${\bar \delta}(a) \neq {\bar \delta}(b)$, and ${\bar \delta}(x) \neq {\bar \delta}(w_1)$.
	By Lemma \ref{lem:samedifferent} (ii), vertices $a$ and $w_1$ have different colors, and therefore
	$b$ and $x$ have different colors as well. 
\end{proof}

\vspace{0.3cm}\begin{lemma}
\label{lem:T8}
If $d_G(z_1)=1$, then transformation $\tau_{8}$ is valid.
\end{lemma}
\vspace{-0.3cm}\begin{proof} 
Suppose $G'$ contains an induced $S_{2,2,2}$, and denote this $S_{2,2,2}$ by $H$. Since $G$ is $S_{2,2,2}$-free, 
both $y_1$ and $x_2$ belong to $H$. If $y_1$ or $x_2$ has degree $1$ in $H$, then an $S_{2,2,2}$ in $G$ is obtained 
by replacing $y_1$ by $y_2$ or $x_2$ by $a$. Hence, one of $y_1,x_2$ has degree $2$, and the other has degree $3$ in $H$. 
But then an $S_{2,2,2}$ in $G$ is obtained by replacing $y_1$ and one of its neighbors different from $x_2$ by $y_2,b$ (if $d_H(y_1)=2$) or $x_2$ and 
one of its neighbors different from $y_1$ by $a,b$ (if $d_H(x_2)=2$), a contradiction.

\vspace{0.1cm}Let ${\bar \gamma}$ be a $\gamma$-completion. Lemma \ref{lem:simpleforcing} (a) implies that $y_1,y_2 \in B_{\gamma}$, which means that
$z_1\notin B_{\gamma}$ and exactly one of $a,b$ is black. If $a$ is white or both $a$ and $x_2$ are black, then a $\delta$-completion ${\bar \delta}$ is obtained by setting 
${\bar \delta}(v)={\bar \gamma}(v)$ for all vertices $v$ in $G'$. If $a$ is black while $x_2$ is white, then $x_1,z_1,b$ are white, and $z_2$ 
is black. Hence, a $\delta$-completion ${\bar \delta}$ is obtained by changing the color of $z_1$ to black  
and setting ${\bar \delta}(v)={\bar \gamma}(v)$ for all other vertices $v$ in $G'$.

\vspace{0.1cm}Let now ${\bar \delta}$ be a $\delta$-completion. Since $y_1,y_2 \in B_{\gamma}$, we have $\{a,b,z_2\}\cap B_{\gamma}=\emptyset$ and 
at most one of $x_1,x_2,z_1$ is black in ${\bar \delta}$. If $x_1$ or $z_1$ is black, or if $x_1,x_2,z_1$ are white, we 
obtain a $\gamma$-completion by assigning color black to $b,y_2$ and color white to $a,z_2$. 
If $x_2$ is black, 
we obtain a $\gamma$-completion by assigning color black to $a,y_2$ and color white to $b,z_2$. 
\vspace{-0.5cm} \end{proof}

\vspace{0.3cm}\begin{lemma}
\label{lem:T9}
Transformation $\tau_{9}$ is valid.
\end{lemma}
\vspace{-0.3cm}\begin{proof}
Suppose $G'$ contains an induced $S_{2,2,2}$, and denote this $S_{2,2,2}$ by $H$. Since $G$ is $S_{2,2,2}$-free, 
both $v_1$ and $v_5$ belong to $H$. If one of them has degree $1$ in $H$, say $v_1$, then an $S_{2,2,2}$ in $G$ is obtained 
by replacing $v_1$ by $v_4$. Hence one of $v_1,v_5$ has degree 2 in $H$, while the other has degree 3, say $d_H(v_1)=2$ and $d_H(v_5)=3$. But then an $S_{2,2,2}$ in $G$ 
is obtained by replacing $v_1$ and one of its neighbors different from $v_5$ by $v_4,v_3$, a contradiction.

\vspace{0.1cm}Since $(G,\gamma)$ is clean, at most one among $v_2,v_3,v_4$ can belong to $B_{\gamma}$, and 
Lemma \ref{lem:simpleforcing} (f) implies that $v_1\in B_{\gamma}$ (resp. $v_2\in B_{\gamma}$) if and only if $v_4\in B_{\gamma}$ (resp. $v_5\in B_{\gamma})$. 
Hence, if $\{v_1,v_4\}\subseteq B_{\gamma}$ (resp. $\{v_2,v_5\}\subseteq B_{\gamma}$) then $\{v_2,v_5\}\cap B_{\gamma}=\emptyset$ (resp. $\{v_1,v_4\}\cap B_{\gamma}=\emptyset$). 

\vspace{0.1cm}Let now ${\bar \gamma}$ be a $\gamma$-completion. If $v_3$ is white, then $v_1,v_2,v_4,v_5$ are black, 
while if $v_3$ is black, then exactly one of $v_2,v_3$ and exactly one of $v_1,v_5$ is black. 
In both cases, we obtain a $\delta$-completion ${\bar \delta}$ by setting ${\bar \delta}(v)={\bar \gamma}(v)$ for all vertices $v$ in $G'$.

\vspace{0.1cm}Finally, let ${\bar \delta}$ be a $\delta$-completion. 
If exactly one of $v_1,v_5$ is black, say $v_1$, then $v_2\notin B_{\gamma}$ 
and we obtain a $\gamma$-completion by assigning color black to $v_3,v_4$ and color white to $v_2$. 

Suppose now that both $v_1,v_5$ are black. We show that $v_3\notin B_{\gamma}$. Assume by contradiction that $v_3\in B_{\gamma}$. 
Then none of $v_1,v_2$ belongs to $B_{\gamma}$ and Lemma (2.2) (a) implies $d_G(v_1)>1$. 
If $d_G(v_1)=2$, Lemma (2.2 (f) implies that $v_1$ has a neighbor $w\neq v_2$ in $B_{\gamma}$, which means that ${\bar \delta}$ is not feasible (since
$w,v_5$ are two black neighbors of $v_1$), a contradiction.
Hence, $d_G(v_1)\geq 3$, and Lemma (2.2) (p) implies that $N(v_1)\setminus \{v_2\}$ contains two adjacent vertices $w,w'$ (else $v_1\in B_{\gamma}$). But Lemma (2.1) (ii) then implies
that $v_1,v_5$ have different colors, a contradiction.
So $v_3\notin B_{\gamma}$, and we obtain a $\gamma$-completion by assigning color black to $v_2,v_4$ and color white to $v_3$.

%
%
\end{proof}


\section{Irreducible graphs}
\label{sec:Irr}


\begin{definition}
We say that a pair $(G,\gamma)$ is {\sc irreducible} if it is clean and none of the reductions $\rho_1,\dots,\rho_8$ and transformations $\tau_1,\dots,\tau_9$ can be applied to $G$. 
\end{definition}

\begin{lemma}\label{lem:outsideTriangle}
	Let $(G,\gamma)$ be an irreducible pair. If $G$ is $S_{2,2,2}$-free, then $\Delta(G)\leq 4$ and every vertex of
	degree 4 belongs to a triangle.
\end{lemma}
\vspace{-0.3cm}\begin{proof}
Assume $\Delta(G)\geq 4$ and let $v$ be a vertex of maximum degree in $G$. Lemma \ref{lem:cleanPairProperties} (a) implies that at most two vertices in $N(v)$ may be adjacent.
Hence, at least $d_G(v)-2$ neighbors of $v$ are 
isolated vertices in the subgraph induced by $N(v)$. If $v\notin B_{\gamma}$ then if follows from Lemma \ref{lem:simpleforcing} (p)  
that $\Delta(G)=4$ and that $v$ belongs to a triangle.
It is therefore sufficient to prove that $v$ cannot belong to $B_{\gamma}$. So assume to the contrary that $v\in B_{\gamma}$. It then follows from Lemma \ref{lem:samedifferent} (ii) that $v$ does not belong to a triangle.
Consider four neighbors $w_1,w_2,w_3,w_4$ of $v$. 
If a neighbor $w$ of $v$ has degree at least $3$, then Lemmas \ref{lem:simpleforcing} (p) and \ref{lem:samedifferent} (ii)
	imply that $w \in W_{\gamma} \cup B_{\gamma}$, which is a contradiction to the assumption that 
	$(G,\gamma)$ is clean. Hence, $d_G(w) \leq 2$ for every $w \in N(v)$. Also, we know from
	Lemma \ref{lem:simpleforcing} (e) that at most one vertex in $N(v)$ has degree 1. Without loss of
	generality, assume $d_G(w_i) = 2$ for $i=1,2,3$ and $d_G(w_4) \leq 2$. 
	For $i = 1, \ldots, 4$, let $u_i$ be the second neighbor of $w_i$ different from $v$, 
	if any. Since $(G,\gamma)$ is clean and $v \in B_{\gamma}$, we know that $w_i,u_i\not\in B_{\gamma}$ for $i=1,\ldots4$. 
	Therefore, Lemma \ref{lem:samedifferent} (i)
	implies that no two vertices in $N(v)$ can have a common neighbor $w$ different from
	$v$, which means that all $u_i$ are distinct.
	
	Suppose that $u_i$ is adjacent to $u_j$ and $u_k$, where $i,j,k$ are three distinct 
	indices. Lemma \ref{lem:simpleforcing} (k) implies that $u_j$ is not adjacent to $u_k$, else at
	least one vertex in $N(v)$ belongs to $W_{\gamma}$. Since $u_i \notin B_{\gamma}$, we know from Lemma \ref{lem:simpleforcing} (p) that $u_i$ has a fourth neighbor $y\neq w_i,u_j,u_k$ adjacent to one of $u_j, u_k$, say $u_j$. Note that Lemma \ref{lem:cleanPairProperties} (a) implies that $y$
	is not adjacent to both $u_j,u_k$. But then Lemma \ref{lem:simpleforcing} (j) implies that $w_j$
	belongs to $W_{\gamma}$, a contradiction. In summary, every $u_i$ is adjacent to at most 
	one other vertex $u_j$.
	
	Suppose $d_G(w_4) =2$. Since $G$ is $S_{2,2,2}$-free, we have $d_G(v)\leq 5$ and the fifth neighbor $w_5$ of $v$, if any, has degree 1.
		Also, without loss of generality, we may assume that $u_1$ is adjacent to $u_2$ while
		$u_3$ is adjacent to $u_4$. Since $u_1$ and $u_2$ do not belong to $B_{\gamma}$,
		it follows from Lemma \ref{lem:simpleforcing} (f) that $d_G(u_1), d_G(u_2) \geq 3$. Lemma \ref{lem:simpleforcing} (p) implies
		then that both $u_1$ and $u_2$ have at least two adjacent neighbors. It then
		follows from Lemma \ref{lem:simpleforcing} (j) that $u_1$ and $u_2$ have a common neighbor $x$,
		otherwise $w_3$ and $w_4$ would belong to $W_{\gamma}$. Similarly, $u_3$ and 
		$u_4$ have a common neighbor $y$. Notice that  
		Lemma \ref{lem:cleanPairProperties} (a) implies $x \neq y$.
		Also, Lemma \ref{lem:simpleforcing} (m) implies that $x$ is not adjacent to $u_3,u_4$, else $x\in W_{\gamma}$.
		Similarly, $y$ is not adjacent to $u_1,u_2$. 
		But this contradict the irreducibility of $(G,\gamma)$ since $\rho_3$ or $\tau_{1}$ can be applied if $d_G(v)=5$, and $\rho_4$ or $\tau_{2}$ can
		be applied if $d_G(v)=4$.
		
	We can therefore suppose that $d_G(v)=4$ and $d_G(w_4) =1$. Since $G$ is $S_{2,2,2}$-free, without
		loss of generality, we may assume that $u_1$ is adjacent to $u_2$. As was the case  when $w_4$ had two neighbors in $G$, we know that $u_1,u_2$ have a common neighbor $x$ that is not adjacent to $u_3$.
		But it then follows that $(G,\gamma)$ is not irreducible since $\tau_{3}$ can be applied, a contradiction. 
\end{proof}

\begin{lemma}\label{lem:insideTriangle}
	Let $(G,\gamma)$ be an irreducible pair. If $G$ is $S_{2,2,2}$-free, then every vertex of
	degree four belongs to a unique triangle and the two other vertices of the triangle have
	degree two.
\end{lemma}
\vspace{-0.3cm}\begin{proof}
	Let $a$ be a vertex of degree four. It follows from Lemmas \ref{lem:outsideTriangle} and \ref{lem:cleanPairProperties} (a)
	that $a$ belongs to a exactly one triangle. Denote by $b,c$ the two other vertices of 
	this triangle. Let $w_1$ and $w_2$ be the two neighbors 
	of $a$ different from $b$ and $c$.
	
	Lemma \ref{lem:cleanPairProperties} (a) implies that $w_1, w_2$ are 
	nonadjacent to $b,c$ and that $w_1$ is nonadjacent to $w_2$. By Lemma \ref{lem:samedifferent} (ii)
	$w_1$ and $w_2$ have the same color. Moreover, there are 
	vertices $u_1$ and $u_2$ such that $N(u_1)\cap \{a,b,c,w_1,w_2,u_2\}=\{w_1\}$ and $N(u_2)\cap \{a,b,c,w_1,w_2,u_1\}=\{w_2\}$, 
	else Lemma \ref{lem:simpleforcing} (r) and \ref{lem:samedifferent} (ii) 
	would imply that $a\in B_{\gamma}$ and $\{w_1,w_2\}\subseteq W_{\gamma}$, and $(G,\gamma)$ would not be clean.
	
	Suppose to the contrary, that at least one of the vertices $b$ or $c$ has degree at least 3, say $b$. 
	Let $x\neq a,c$ be a third neighbor of $b$.
		\vspace{0.2cm}
	
	\noindent
\textbf{Case 1}. \textit{$(N(w_1)\cup N(w_2))\cap (N(b)\cup N(c))=\{a\})$}
	
	Since $G$ is $S_{2,2,2}$-free, $x$ must be adjacent to $u_1$ or $u_2$. 
	Lemma \ref{lem:cleanPairProperties} (b) implies that $x$ is adjacent exactly to
	one of them, say $u_1$, that $u_1$ is black and does not belong to a triangle, and that $x,w_1$ have no other neighbors.
	Hence, Lemma \ref{lem:outsideTriangle} implies that $d_G(u_1)\leq 3$ (otherwise $w_1,v\in W_{\gamma}$). 
	Suppose $b$ has a fourth neighbor $y\neq a,c,x$. Similarly to $x$,
	vertex $y$ is adjacent to exactly one of the vertices $u_1$ and $u_2$. 
	If $y$ is adjacent to $u_1$, then Lemma \ref{lem:samedifferent} (i) implies that $\{x,y\}\in W_{\gamma}$, while if $y$ is adjacent to $u_2$ then  
	Lemma \ref{lem:simpleforcing} (n) implies $w_1,w_2 \in W_{\gamma}$, a contradiction. Hence both $b$ and $c$ have at most 3 neighbors.
	
	Suppose $c$ also has a third neighbor $y\neq a,b$. 
	\begin{enumerate}
\vspace{-0.1cm}		\item \textit{Both $x$ and $y$ are adjacent to $u_1$}. By Lemma \ref{lem:cleanPairProperties} (b) $y$ has no more neighbors. This contradicts the irreducibility of $(G,\gamma)$ since $\rho_2$ can be applied.
		
\vspace{-0.1cm}		\item \textit{$x$ is adjacent to $u_1$ and $y$ is adjacent to $u_2$}.
		As for vertex $u_1$, Lemma \ref{lem:cleanPairProperties} (b) and \ref{lem:outsideTriangle} imply that $u_2\in B_{\gamma}$ and $d_G(u_2)\leq 3$.
		As $(G,\gamma)$ is clean, none of the vertices $a,b,c,x,y,w_1,w_2$ 
		is in $B_{\gamma}$. It then follows from Lemma \ref{lem:simpleforcing} (i) that $d_G(u_1)=3$, else $c\in B_{\gamma}$.
		Similarly, $d_G(u_2)=3$. So, let $d\neq x,w_1$ be a third neighbor of $u_1$, and let $e\neq y,w_2$ be a third neighbor of $u_2$.
		Since $u_1, u_2$ both belong to $B_{\gamma}$ and since $(G,\gamma)$ is irreducible, they have no common
		neighbors, and therefore $d$ and $e$ are different. Moreover, $d$ and $e$
		are not adjacent, else Lemma \ref{lem:simpleforcing} (s) implies $a\in B_{\gamma}$, a contradiction.
		But now $\tau_{4}$ can be applied, a contradiction.
\end{enumerate}
	\vspace{-0.2cm}Thus we may assume now that $d_G(c)=2$. Then $d_G(u_1) \neq 2$, else $(G,\gamma)$ is not irreducible since $\rho_8$ can be applied.
	Hence, $u_1$ has a third neighbor $f\neq x,w_1$, and $d_G(f) \leq 2$ by Lemma \ref{lem:cleanPairProperties} (b). Also, 
	Lemma \ref{lem:simpleforcing} (j) implies that $f$ is not adjacent to
	$w_2$, else $f\in W_{\gamma}$. But this contradicts the irreducibility of $(G,\gamma)$ since $\tau_{5}$ can be applied.
	\vspace{0.2cm}
	
	\noindent
	\textbf{Case 2}. \textit{$\mid (N(w_1)\cup N(w_2))\cap (N(b)\cup N(c))\mid \geq 2$}
	
	Assume, without loss of generality, that $x$ is adjacent to $w_1$. By 
	Lemma \ref{lem:cleanPairProperties} (c), we have $c\in B_{\gamma}$ and $d_G(c)=2$.
	First, we show that $x$ is nonadjacent to both $w_2,u_2$. 
	Lemma \ref{lem:cleanPairProperties} (b) excludes the case when $x$ is adjacent 
	to $u_2$, but is not adjacent to $w_2$. Therefore, suppose $x$ is adjacent to 
	$w_2$. In order to avoid a forbidden $S_{2,2,2}$ (induced by
	$x,b,c,w_1,u_1,w_2,u_2$), $x$ must be adjacent to $u_1$ or $u_2$. Lemma \ref{lem:outsideTriangle} implies 
	$d_G(x)\leq 4$ and $x$ is therefore adjacent to exactly one of 
	$u_1$ and $u_2$. By symmetry, we may assume without loss of generality that $x$ is adjacent to $u_1$. 
	By Lemma \ref{lem:cleanPairProperties} (c) $u_1$ has no neighbors different from
	$x,w_1$ and $u_1\in B_{\gamma}$.
	Note that $w_2$ must belong to a triangle, because otherwise
	Lemma \ref{lem:simpleforcing} (p) and Lemma \ref{lem:samedifferent} (ii) would imply $a \in W_{\gamma}$.
	Hence there is a vertex $u_2'$ adjacent to both $w_2$ and $u_2$.
By Lemma \ref{lem:outsideTriangle}, $w_2$ has no other neighbors than $a,x,u_2,u_2'$.
	Moreover, neither $u_2$, nor $u_2'$ has a neighbor outside 
	$\{a,b,c,x,w_1,u_1,w_2\}$. Indeed,
	if say $u_2'$ had such a neighbor $z$, then $\{w_2,u_2',z,x,u_1,a,c\}$ would
	induce an $S_{2,2,2}$. Also, it follows from Lemma \ref{lem:cleanPairProperties} (a) and (c) that at most one
 of $u_2,u_2'$  can be adjacent to $b$ or $w_1$. 
But then $\rho_5$ can be applied and  
	$(G,\gamma)$ is therefore not irreducible, a contradiction.
	
	Now Lemma \ref{lem:simpleforcing} ($\ell$) implies $d_G(x)\geq 3$, and it follows from Lemma \ref{lem:simpleforcing} (p) and 
	Lemma \ref{lem:samedifferent} (i) that $x$ belongs to a triangle $T_1$. Similarly, $w_1$ belongs 
	to a triangle $T_2$. 
		\vspace{-0.2cm}
\begin{enumerate}
		\item \textit{$T_1 \neq T_2$}. Lemma \ref{lem:cleanPairProperties} (a) implies that 
		the triangles have no common vertices. Moreover, by Lemma 
		\ref{lem:cleanPairProperties} (d) there are at most two edges between 
		$T_1$ and $T_2$ and if there are exactly two edges, then they are not adjacent. 
		Let $x,y,z$ be the vertices of $T_1$ and $w_1,u_1,u_1'$ be the vertices of $T_2$.
		Denote by $M$ the set of vertices $\{ a,b,c,x,y,z,w_1,u_1,u_1',w_2,u_2 \}$.

\vspace{-0.1cm}It follows from Lemma \ref{lem:outsideTriangle} that $x$ and $w_1$ have no neighbors outside $M$.
Also, $u_1$ and $u_1'$ have no neighbors outside $M$. Indeed, if say $u_1$ had such a neighbor $r$, then Lemmas 
\ref{lem:cleanPairProperties} (a) and (c) imply that $w_1,a,c,u_1,r,x$ together with $y$ or $z$ induce a $S_{2,2,2}$, a contradiction. 
Moreover, it follows from Lemmas \ref{lem:cleanPairProperties} (a) and (b) that $y$ and $z$ are nonadjacent to both $b,w_2$. 
It then follows from Lemma \ref{lem:cleanPairProperties} (a) and (c) that $\rho_{7}$ can be applied.
Indeed, if $N(u_1')\cap \{b,w_2\}\neq \emptyset$, then $d_G(u_1)=2$ and $u_1'$ is adjacent to at most one of $y,z$, while if
$N(u_1')\cap \{b,w_2\}= \emptyset$, then at most one of $u_1,u_1'$ is adjacent to at most one of $y,z$. 
Hence $(G,\gamma)$ is not irreducible, a contradiction.
		
\vspace{-0.1cm}		\item \textit{$T_1 = T_2$}. Let $x,w_1,s$ be the vertices of $T_1$ (where $s$
		may coincide with $u_1$). Lemma \ref{lem:cleanPairProperties} (a) and (c) implies that $d_G(s)=2$ and $s\in B_{\gamma}$.
		Now $d_G(b)>3$, else $\tau_{7}$ can be applied and $(G,\gamma)$ is not irreducible. Let $y\neq a,c,x$ be the fourth neighbor of $b$. 
Then Lemma \ref{lem:cleanPairProperties} (a) implies that $y$ is not adjacent to $x$. Also, $y$ is adjacent to $w_1$ or $w_2$. 
Indeed, if $y$ is nonadjacent to $w_1,w_2$, then since $G$ is $S_{2,2,2}$-free, it must
		be adjacent to $u_1$ or $u_2$. Lemma \ref{lem:cleanPairProperties} (b) implies
		that $y$ cannot be adjacent to $u_1$. Hence, $y$ is adjacent to
		$u_2$. It follows from Lemma \ref{lem:cleanPairProperties} (b) that $d_G(y)=d_G(w_2)=2$ and $u_2\in B_{\gamma}$. Furthermore, if $d_G(u_2)\geq 3$, say $u_2$ has a neighbor $t\neq y,w_2$, then it follows from Lemma \ref{lem:simpleforcing} (t) that $t$ must be white, a contradiction. Hence, $d_G(u_2)=2$ but this contradicts the irreducibility of $(G,\gamma)$ since $\rho_6$ can be applied.

\vspace{-0.1cm}		If now $y$ is adjacent to $w_1$, then, similarly to $x$, we conclude that $y$
		and $w_1$ belong to a same triangle, which is impossible by
		Lemma \ref{lem:cleanPairProperties} (a). Hence $y$ is adjacent to $w_2$, and similarly to $x$, we know that
		$y$ and $w_2$ have a common neighbor $p$ (possibly equal to $u_2$). Furthermore, $d_G(p)=2$ by Lemma \ref{lem:cleanPairProperties} (c). But then $\tau_{6}$ can be applied, which contradicts the irreducibility of $(G,\gamma)$.
		
	\end{enumerate}
\vspace{-0.8cm}\end{proof}

\vspace{0.3cm}In the remainder of the paper  $\mathcal{T}$ will denote the subset of vertices that belong to a triangle in $G$.
\begin{lemma}
\label{lem:degree3}
Let $(G,\gamma)$ be an irreducible pair where $G=(V,E)$ is $S_{2,2,2}$-free. 
Let $P$ be an induced path in $G$ with edge set $\{v_1v_2,v_2v_3,\ldots,v_{\ell-1}v_{\ell}\}$ ($\ell>1$) and 
with $d_G(v_1)\geq 3$, $d_G(v_{\ell})\geq 3$ and $d_G(v_i)=2$ for $i=2,\ldots,\ell-1$. We have either
\begin{itemize}
\vspace{-0.3cm} \item[(1)] $\ell=2$ and both $v_1,v_{\ell}$ belong to $\mathcal{T}$, or
\vspace{-0.3cm} \item[(2)] $\ell=3$ and exactly one of $v_1,v_{\ell}$ belongs to $\mathcal{T}$, or
\vspace{-0.3cm} \item[(3)] $\ell=4$ and none of $v_1,v_{\ell}$ belongs to $\mathcal{T}$.
\end{itemize}
\end{lemma}
\vspace{-0.3cm}\begin{proof}
We necessarily have $\ell\leq 4$, otherwise $\tau_9$ can be applied and hence $(G,\gamma)$ would not be irreducible.
\begin{itemize}
\vspace{-0.3cm}\item If $\ell=2$ then at least one of $v_1,v_{2}$ belongs to $\mathcal{T}$ else Lemma \ref{lem:simpleforcing} (p) implies that the two adjacent vertices $v_1,v_{2}$ belong to $B_{\gamma}$. If only one of $v_1,v_{2}$ belongs to $\mathcal{T}$, say $v_1$, then Lemma \ref{lem:simpleforcing} (p) and Lemma \ref{lem:samedifferent} (ii) imply that $v_{2}\in B_{\gamma}$ and $v_1\in W_{\gamma}$, a contradiction.
\vspace{-0.3cm}\item Suppose $\ell=3$. Since $v_1$ and $v_{3}$ are at distance 2, at least one of  $v_1,v_{3}$ does not belong to $B_{\gamma}$, say $v_1$. Then  Lemma \ref{lem:simpleforcing} (g) and (p) imply that $v_1\in \mathcal{T}$ and $v_{3}\notin \mathcal{T}$.  
\vspace{-0.3cm}\item Suppose $\ell=4$  and, without loss of generality, assume $v_1\in \mathcal{T}$. Then Lemma \ref{lem:simpleforcing} (g) 
implies $v_3\in B_{\gamma}$. Now, either  $v_{4}\in \mathcal{T}$ and  Lemma \ref{lem:simpleforcing} (g) implies that $v_2\in B_{\gamma}$, or $v_{4}\notin \mathcal{T}$ and Lemma \ref{lem:simpleforcing} (p) implies that $v_{4} \in B_{\gamma}$. Hence $G$ contains two adjacent black vertices, a contradiction.
\end{itemize}
\vspace{-0.4cm}\end{proof}

\vspace{0.3cm}\begin{lemma}
\label{lem:Cp+1tree} Let $(G,\gamma)$ be an irreducible pair where $G$ is $S_{2,2,2}$-free. 
If $G[V\setminus \mathcal{T}]$ contains an induced cycle, then $\gamma$ is not completable. 
\end{lemma}

\vspace{-0.2cm}\begin{proof} 
Let $C$ be an induced cycle in $G[V\setminus \mathcal{T}]$ with edge set $\{v_1v_2,\cdots,v_{k-1}v_k,v_kv_1\}$. 
Note that $k\geq 4$ since $G[V\setminus \mathcal{T}]$ contains no triangle. If $C$ is a connected component of $G$, then $k\leq 5$, otherwise $\tau_9$ can be applied and $(G,\gamma)$ would not be irreducible. It follows from Lemma \ref{lem:simpleforcing} ($\ell$) that $k\neq 4$. Hence, $k=5$ and it is then easy to observe that $C$ (and thus $G$) does not admit any feasible complete coloring.

So assume $C$ contains at least one vertex which has degree 3 in $G$, say $v_1$. We know from Lemma \ref{lem:degree3} that no vertex in $C$ has a neighbor 
in $G[\mathcal{T}]$. Hence $v_1$ has a neighbor $w\neq v_2,v_{k}$ in $G[V\setminus \mathcal{T}]$. It then follows from 
Lemma \ref{lem:simpleforcing} (p) that $v_1\in B_{\gamma}$. Also, $w$ has no other neighbor on $C$ else this neighbor would also belong to $B_{\gamma}$ 
and would be at distance 2 from $v_1$. Also, $k\geq 5$, else Lemma \ref{lem:samedifferent} (i) would imply that $v_2,v_4 \in W_{\gamma}$.
Now, Lemma \ref{lem:simpleforcing} (p) implies $d_G(v_2)=d_G(v_3)=d_G(v_{k-1})=d_G(v_k)=2$, else $B_{\gamma}$ would contain two vertices at distance at most 2. 
Since $v_3\notin B_{\gamma}$, we know from Lemma \ref{lem:simpleforcing} (f) that $k\geq 6$. Hence $d_G(w)=1$, otherwise $w$ would 
have a second neighbor $x\neq v_1$ and vertices $v_1,v_2,v_3,v_{k-1},v_k,w,x$ would induce an $S_{2,2,2}$ in $G$. 
Also, $v_4$ has a neighbor $w'\neq v_3,v_5$ in $G[V\setminus \mathcal{T}]$, else $\tau_9$ can be applied and $(G,\gamma)$ would not be irreducible.
Using the same arguments as for $w$, we obtain that $d_G(w')=1$. Now $k>6$ by Lemma \ref{lem:simpleforcing} (q). 
But then $\tau_{8}$ can be applied, which contradicts the irreducibility of $(G,\gamma)$. 
\vspace{-0.8cm}\end{proof}

\vspace{0.5cm}\begin{lemma}
\label{lem:Cp+1claw}  Let $(G,\gamma)$ be an irreducible pair where $G$ is $S_{2,2,2}$-free. 
Then every connected component of $G[V\setminus \mathcal{T}]$ is a claw whose center has exactly one neighbor of degree 1 in $G$. 
\end{lemma}

\vspace{-0.3cm}\begin{proof} Let $H$ be a connected component of $G[V\setminus \mathcal{T}]$. By Lemma \ref{lem:Cp+1tree}, 
$H$ is a tree. If $H$ contains only one vertex, then it follows from Lemmas \ref{lem:simpleforcing} (a) and \ref{lem:samedifferent} (ii) that $u\in W_{\gamma}$, a contradiction. So $H$ contains at least two vertices. Let $u$ be a vertex in $H$. If $d_H(u)=0$, it follows from Lemma \ref{lem:degree3} that $u$ has at most one neighbor in $\mathcal{T}$. If $1\leq d_H(u)\leq 3$, then it follows from Lemmas \ref{lem:simpleforcing} (p) and Lemma \ref{lem:samedifferent} (i) that $u$ has no neighbor in $\mathcal{T}$.

\vspace{0.3cm}\noindent Claim 1. {\it If $d_H(u)=1$ for some vertex $u$ of $H$, then its neighbor in $H$ is the center of a claw.}  

 Let $v$ be the neighbor of $u$ in $H$. We first prove that $v\in B_{\gamma}$. If $d_G(u)=1$ or $d_G(v)=3$, then $v\in B_{\gamma}$ 
by Lemma \ref{lem:simpleforcing} (a) and (p). We show that the other cases are impossible. Assume $d_G(u)=2$ and $d_G(v)\leq 2$. 
Then $d_G(v)=2$, else it follows from Lemma \ref{lem:simpleforcing} (a) and (g) that $u$ and $v$ are two adjacent vertices in $B_{\gamma}$. 
Let $w\neq u$ be the second neighbor of $v$. We know from Lemma \ref{lem:degree3} that $d_G(w)< 3$. 
Now, $\rho_1$ (if $d_G(w)=1$) or $\tau_9$ (if $d_G(w)=2$) can be applied, which contradicts the irreducibility of $(G,\gamma)$.

 Let $P$ be a path in $H$ with edge set $\{v_1v_2,\cdots,v_{k-1}v_k\}$, $u=v_1$, $d_G(v_k)\neq 2$ and $d_G(v_i)=2$, $i=2,\cdots,k-1$. 
Since $(G,\gamma)$ is irreducible, we have $k\leq 4$, otherwise $\tau_9$ can be applied. 
Also, $\{v_1,v_3,v_4\}\cap B_{\gamma}=\emptyset$ since $v=v_2\in B_{\gamma}$.
If $k=2$, then $d_G(v_2)\neq 1$ otherwise $v_1\in B_{\gamma}$ and thus we would have two adjacent vertices in $B_{\gamma}$. Hence, $d_G(v_2)=3$ and so $v_2$ is the center of a claw.
We finally show that $k$ cannot be equal to 3 or 4.  
If $k=4$ then Lemma \ref{lem:simpleforcing} (a) and (p) imply $v_3\in B_{\gamma}$ (if $d_G(v_4)=1$) or $v_4\in B_{\gamma}$ (if $d_G(v_4)=3$), a contradiction. 
If $k=3$ then $d_G(v_3)=1$, else Lemma \ref{lem:simpleforcing} (p) implies $v_3\in B_{\gamma}$. But now Lemma \ref{lem:simpleforcing} (a) implies that $v_3\in B_{\gamma}$ and hence $(G,\gamma)$ is not irreducible, and Claim 1 is proven.

\vspace{0.3cm}

 \noindent Claim 2. {\it If $v$ is the center of a claw in $H$, then exactly one of its neighbors has degree 1 in $G$.}

 It follows from Lemmas \ref{lem:degree3} and \ref{lem:simpleforcing} (p) that $v\in B_{\gamma}$ and no neighbor of $v$ 
can be in $G[\mathcal{T}]$ or a center of a claw. Hence, all neighbors of $v$ have degree at most 2 in $G$ and none of them belongs to $B_{\gamma}$. 
We know from 
Lemma \ref{lem:simpleforcing} (e) that at most one neighbor of $v$ can have degree 1 in $G$. So assume, by contradiction, 
that the three neighbors $u_1,u_2,u_3$ of $v$ are of degree 2 in $G$, and let $w_1,w_2,w_3$ be their respective second neighbor. 
Note that $w_1,w_2,w_3$ are all distinct by Lemma \ref{lem:samedifferent} (i). If $w_1,w_2,w_3$ induce a triangle, 
then $\rho_2$ can be applied, and $(G,\gamma)$ is therefore not irreducible, a contradiction.
Since $G$ is $S_{2,2,2}$-free, at least two of $w_1,w_2,w_3$, say $w_1,w_2$, are adjacent. Since $H$ is a tree, at least one of 
$w_1,w_2$ belongs to $\mathcal{T}$, say $w_1$. If $w_2\notin \mathcal{T}$, then $d_G(w_2)=2$ (otherwise $w_2\in B_{\gamma}, w_1\in W_{\gamma}$) and Lemma \ref{lem:simpleforcing} (g) 
implies $u_2\in B_{\gamma}$, a contradiction. So $w_2\in \mathcal{T}$. If $w_1,w_2$ belong to distinct triangles, then it follows from Lemma \ref{lem:simpleforcing} (j) that $u_3\in W_{\gamma}$. Hence, $\mathcal{T}$ contains a vertex $y$ adjacent to $w_1$ 
and $w_2$. It then follows from Lemma \ref{lem:insideTriangle} that $d_G(w_1)=d_G(w_2)=3$. 
Hence, Lemma \ref{lem:simpleforcing} (i) implies that one of $u_1,u_2$ belongs to $W_{\gamma}$, a contradiction, and Claim 2 is proven.

\vspace{0.3cm} Let $\{v_1v_2,\cdots,v_{k-1}v_k\}$ be the edge set of a longest path in $H$. It follows from Claim 1 that $k\geq 3$ and 
$v_2$ and $v_{k-1}$ are centers of a claw. If $k>3$, then we know from Lemma \ref{lem:degree3} that $k\geq 6$ and that $v_5$ is the center 
of a claw. Let $v_2'\neq v_1,v_3$ and $v_5'\neq v_4,v_6$ denote the third neighbors of $v_2$ and $v_5$, respectively. We know from Claim 2 
that one of $v_1,v_2'$ and one of $v_6,v_5'$ 
has degree one in $G$, which means that $\tau_{8}$ can be applied, and $(G,\gamma)$ is therefore not irreducible, 
a contradiction. So $k=3$ and $H$ is a claw.
\end{proof} 


\section{Dominating induced matching in $S_{2,2,2}$-free graphs}
\label{sec:solution}

The main results of the previous sections can be summarized as follows. Let $G$ be an $S_{2,2,2}$-free graph and let $(G, \gamma)$ be an irreducible pair. Then

\begin{itemize}
\vspace{-0.3cm}		\item every vertex in $\mathcal{T} \cap B_{\gamma}$ has degree two;
\vspace{-0.3cm}		\item every maximal clique in $G$ has two or three vertices;
\vspace{-0.3cm}		\item every vertex of $G$ belongs to at most one triangle;
\vspace{-0.3cm}		\item every vertex of $G$ has degree at most four;
\vspace{-0.3cm}		\item every vertex $r$ of degree four belongs to a triangle $T$, in which
		the other two vertices $p$ and $q$ have degree two
		(see Figure \ref{fig:clawlink} (a));
\vspace{-0.3cm}		\item every connected component of $G[V\setminus \mathcal{T}]$ is a claw.
		Moreover, if $x,a_1,a_2,a_3$ are vertices of a connected component $C$ of
		$G[V\setminus \mathcal{T}]$, such that $d_C(x)=3$, then 
		$d_G(x)=3$, $d_G(a_1)=d_G(a_3)=2$, $d_G(a_2)=1$ and the neighbors of
		$a_1$ and $a_3$ different from $x$ belong to different triangles in $G$
		(see Figure \ref{fig:clawlink} (b)).
	\end{itemize}

\begin{figure}
\begin{center}
\scalebox{0.16}{\includegraphics{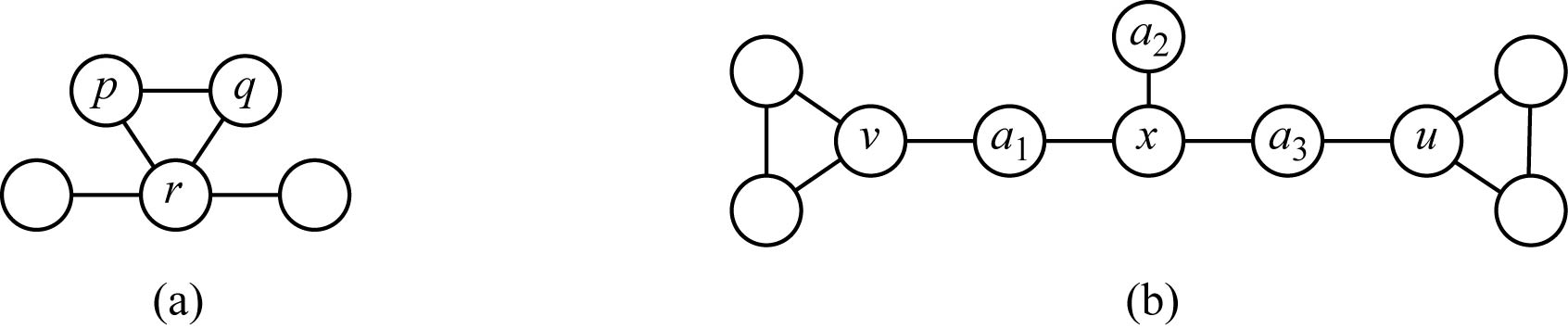}}\caption{Irreducible graph structures.}
\label{fig:clawlink}
\end{center}
\end{figure}

Let $(G, \gamma)$ be an irreducible pair.
Let $C^1, \ldots, C^k$ be connected components of $G[V\setminus \mathcal{T}]$ and let
$x^i,a_1^i, a_2^i,a_3^i$ be vertices of $C^i$, such that $d_G(x^i)=3$,
$d_G(a_1^i) = d_G(a_3^i) = 2$ and $d_G(a_2^i)=1$. Denote by $v^i$ and $u^i$ the 
neighbors of $a_1^i$ and $a_3^i$ in $G$, respectively.
Also, let $T^1, \ldots, T^s$ be the triangles in $G$ which contain a vertex of
degree four and let $r^i,p^i,q^i$ be the vertices of $T^i$, such that $d_G(r^i)=4$ and 
$d_G(p^i)=d_G(q^i)=2$. Let $\mathcal{S}=\{ p^1, q^1, \ldots, p^s,q^s\}$ the set of vertices of degree two
in triangles $T^1, \ldots, T^s$. Let $G'$
be the subgraph of $G$ induced by
$V' = \mathcal{T} \setminus (\mathcal{S} \cup B_{\gamma})$.

Define a family $\mathcal{K}_{(G, \gamma)}$ of subsets of vertices of $G$ in the following
way:
\begin{enumerate}
\vspace{-0.3cm}	\item introduce into $\mathcal{K}_{(G, \gamma)}$ every maximal clique of $G'$ of size
	strictly greater than one;
	
\vspace{-0.3cm}	\item for every connected component $C^i$ of $G[V\setminus \mathcal{T}]$ introduce 
	$\{ v^i, u^i, a_2^i \}$ into $\mathcal{K}_{(G, \gamma)}$.
\end{enumerate}

\noindent
Using the definition and the above properties of irreducible pairs $(G,\gamma)$  it is easy to check that 
$\mathcal{K}_{(G, \gamma)}$ satisfies the following properties:
\begin{enumerate}
\vspace{-0.3cm}	\item[(1)] every set in $\mathcal{K}_{(G, \gamma)}$ has two or three vertices;
\vspace{-0.3cm}	\item[(2)] every vertex of $G$ belongs to at most two sets of $\mathcal{K}_{(G, \gamma)}$.
\end{enumerate} 

Let $M = V' \cup \{ a_2^1, \ldots, a_2^k\}$ .
\begin{lemma}
\label{lem:dim_sets_reduction}
Let $(G, \gamma)$ be an irreducible pair.
Then $\gamma$ is completable if and only if there exists a set $H \subseteq M$
such that $|H \cap K| = 1$ for every $K\in \mathcal{K}_{(G, \gamma)}$.
\end{lemma}

\vspace{-0.3cm}\begin{proof}
	Let ${\bar \gamma}$ be a $\gamma$-completion of $G$ and let
	$$
	H = (V' \cap W_{\bar \gamma}) \cup ( \{ a_2^1, \ldots, a_2^k\}\cap B_{\bar \gamma}). 
	$$
	
	Clearly, $H$ is a subset of $M$. Therefore we only need to show that every 
	$K\in \mathcal{K}_{(G, \gamma)}$ contains exactly one element in $H$.
	\begin{enumerate}
		\item \textit{Let $K = \{x,y\}$}. Since $K$ has two elements, the definition
		of $\mathcal{K}_{(G, \gamma)}$ implies that $K$ is a 
		maximal clique in $G'$.
		
\vspace{-0.1cm}		If $x,y$ belong to the same triangle $T$ in $G$, then the third vertex $z$ of $T$
		does not belong to $G'$, and therefore $z \in B_{\gamma}$. It means that ${\bar \gamma}$
		assigns color white to exactly one of the vertices $x, y$ and therefore 
		$|H \cap K| = 1$.
	
\vspace{-0.1cm}		If $x,y$ belong to different triangles in $G$, then by Lemma 2.1 (ii)
		${\bar \gamma}$ assigns different colors to $x$ and $y$ and therefore exactly one of 
		them belongs to $H$.
		
		\item \textit{Let $K = \{x,y,z\}$}.
		
\vspace{-0.1cm}		If $K$ is a maximal clique in $G'$, then $K$ induces a triangle in $G$. 
		Since ${\bar \gamma}$ is
		a feasible complete coloring, exactly one of the vertices $x,y,z$ is in
		$W_{\bar \gamma}$ and therefore in $H$.
		
\vspace{-0.1cm}		If $K$ is not a clique in $G'$, then $K = \{ v^i, u^i, a_2^i \}$ for some 
		$i \in \{1, \ldots, k \}$. Note that $v^i, u^i$ cannot both be white, because
		otherwise $a_1^i$ and $a_3^i$ would be two black neighbors of the black vertex $x^i$. Thus, either both are black implying that $a_2^i$ is black or one is black and the other one is white implying that $a_2^i$ is white. Hence $|H \cap K| = 1$.
	\end{enumerate}
	
	Let now $H$ be a set satisfying the conditions of the lemma and let $\delta$ be
	a coloring of $G$ defined in the following way:
	\begin{enumerate}
\vspace{-0.3cm}		\item[(1)] $\delta$ assigns color black to every $x \in B_{\gamma}$;
\vspace{-0.3cm}		\item[(2)] $\delta$ assigns color white to every $x \in V' \cap H$;
\vspace{-0.3cm}		\item[(3)] $\delta$ assigns color black to every $x \in V' \setminus H$;
\vspace{-0.3cm}		\item[(4)] for every $i = 1, \ldots, s$, define $\delta(p^i)$ and $\delta(q^i)$
		(if not yet defined) in such a way that $T^i$ contains exactly one white vertex. 
		(Note that one of the vertices $p^i$ and $q^i$ may already have a color assigned 
		by $\delta$ if this vertex belongs to $B_{\gamma}$);
\vspace{-0.3cm}		\item[(5)] for every $i = 1, \ldots, k$, $\delta$ assigns color black to $a_2^i$ 
		and color white to $a_1^i$ and $a_3^i$, if $a_2^i \in H$; $\delta$ assigns color 
		white to $a_2^i$, $a_1^i$ and color black to $a_3^i$, if $u^i \in H$; 
		$\delta$ assigns color white to $a_2^i$, $a_3^i$
		and color black to $a_1^i$, if $v^i \in H$. Note that $x^i$ is already assigned color
		black, since it belongs to $B_{\gamma}$ (by Lemma 2.2 (p)).
	\end{enumerate}
	
	Clearly, $\delta$ is a complete coloring extending $\gamma$. We therefore 
	only need to show that $\delta$ is feasible. 
	To this end we prove that $W_{\delta}$
	is an independent set and $B_{\delta}$ induces a 1-regular subgraph (i.e. a graph in which all vertices have degree exactly one) in $G$.
	
	\begin{enumerate}
		\item \textit{$W_{\delta}$ is an independent set in $G$}.
		The definition of $\delta$ implies that every white vertex $x \in V \setminus V'$
		has no white neighbors.
		Let $x,y$ be two white vertices in $V'$. By items (2) and (3) of the definition of
		$\delta$, vertices $x$ and $y$ belong to $H$, and therefore 
		they are not adjacent, since otherwise $x,y$ would belong 
		to a maximal clique of
		size at least two in $G'$, which contradicts the assumption that $H$ intersects
		every  nontrivial maximal clique of $G'$ in exactly one vertex.
		
		\item \textit{$B_{\delta}$ induces a 1-regular graph in $G$}. By (5) every
		black vertex in $V\setminus \mathcal{T}$ has exactly one black neighbor and this neighbor
		belongs to $V\setminus \mathcal{T}$ as well. 
		Therefore it is sufficient to show that every black vertex 
		in $\mathcal{T}$ has exactly one black neighbor in $\mathcal{T}$. 
		
		Firstly, we show that a black vertex of a triangle $T$ has exactly one black
		neighbor in $T$ or equivalently that every triangle $T$ has exactly two 
		black vertices. 
		If $T$ is one of the triangles $T^1, \ldots, T^s$ then this is
		provided by (4). If $T$ does not contain a vertex of degree four, but has a vertex
		$x \in B_{\gamma}$, then $x \notin V'$ and the two other vertices of $T$ form a
		maximal clique in $G'$ and hence exactly one of them is black. Otherwise, the
		vertices of $T$ form a maximal clique in $G'$ and exactly two of them are black.
		
		Now let $x$ be a black vertex of a triangle $T$ and suppose $x$ has a neighbor 
		$y \in \mathcal{T}$ outside $T$. Since every vertex
		of $G$ belongs to at most one triangle, $\{ x, y \}$ forms a maximal clique 
		in $G'$, and therefore $y$ must be white.
	\end{enumerate}
\vspace{-0.8cm}\end{proof}

Lemma \ref{lem:dim_sets_reduction} reduces the Dominating Induced Matching Problem 
in $S_{2,2,2}$-free graphs to the following.
\vspace{0.2cm}

\noindent
\textbf{Problem A}. We are given a finite set $S$ and a family 
$\mathcal{F} = \{ A_i | i \in I \}$ of subsets $A_i \subseteq S$ such that each
element of $S$ appears in at most two members of $\mathcal{F}$. We have
to find (if it exists) a subset $C \subseteq S$ such that $|C \cap A_i|=1$ for each $i \in I$.

\vspace{0.3cm}We now prove how to solve this problem in polynomial time.
Let $G=(V,E)$ be a graph and $M$ be a matching in $G$. We say that $M$ \textit{saturates}
a set $U \subseteq V$ if every vertex in $U$ is incident with an edge in $M$. The vertices
in $U$ are called \textit{saturated} (by $M$) and vertices not incident with any edge of $M$
are \textit{unsaturated}. A path in $G$ which contains, alternately, edges in $M$ and 
$\overline{M}$, where $\overline{M} = E \setminus M$, is called an \textit{alternating path} with
respect to $M$. An alternating path that starts and ends in unsaturated vertices is called
an \textit{augmenting path} with respect to $M$.

Without loss of generality, we may assume that in Problem A for any $i,j \in I$ the sets 
$A_i$ and $A_j$ have at most one common element. Indeed, all elements of 
$A_i \cap A_j$ belong to exactly the same members in $\mathcal{F}$ and at most one of 
the elements will appear in $C$. Therefore we can remove all but one element of 
each intersection $A_i \cap A_j$.

With a given instance of Problem A we may associate a graph $G=(V,E)$, where 
$V = \{ a_i | i \in I \}$ and two different vertices $a_i$ and $a_j$ are adjacent if and only if 
$A_i \cap A_j \neq \emptyset$. Now let $U \subseteq V$ be the set of vertices $a_i$ of $G$
such that each element in $A_i$ belongs to exactly two members of $\mathcal{F}$.
Consider the following problem.
\vspace{0.3cm}

\noindent
\textbf{Problem B}. Given a graph $G = (V,E)$ and a subset $U \subseteq V$ find 
(if it exists) a matching which saturates all vertices in $U$.

\begin{lemma}
\label{lem:A_equiv_B}
	Problem A has a solution if and only if Problem B has a solution.
\end{lemma}

\vspace{-0.2cm}\begin{proof}
	Assume we have a solution for Problem A, i.e. there exists a subset $C \subseteq S$
	satisfying $|C \cap A_i| = 1$ for all $i \in I$. Let $M$ be the set of edges $a_ia_j$ 
	in $G$, such that the common element of $A_i$ and $A_j$ belongs to $C$. Clearly,
	no two edges in $M$ can have a common endpoint $a_i$, otherwise $C$ would
	contain at least two elements of $A_i$, which is a contradiction. In other words, the edges
	of $M$ form a matching in $G$. Let now $a_i$ be a  vertex of $G$ such that all
	elements of $A_i$ belong to exactly two members of $\mathcal{F}$. In particular,
	the element $e \in C \cap A_i$ corresponds to some edge
	$a_ia_j \in M$. Hence $M$ is a matching in $G$ which saturates all vertices in $U$.
	
	Conversely, let $M$ be a matching in $G$ which saturates all vertices in $U$. Let
	$C = \emptyset$. For each edge $a_ia_j \in M$ we introduce the element 
	$e \in A_i \cap A_j$ into $C$. This gives us a subset $C \subseteq S$ with 
	$|C \cap A_i| = 1$ for each $A_i$ such that all elements of $A_i$ are in exactly
	two members of $\mathcal{F}$. Consider the members of $\mathcal{F}$ which have
	no common element with $C$.
	Each such subset $A_i$ contains some element $e_i$ such that 
	$e_i \notin A_j$ for every $j \in I, j \neq i$. 
	By introducing one of these elements into $C$ for every $A_i$ we obtain a subset with $|C \cap A_i| = 1$ for all $i \in I$.
\end{proof}

\vspace{0.3cm}Further we show that Problem B, and therefore Problem A, can be solved in polynomial
time. We start with two auxiliary statements.

\begin{lemma} \label{lem:mmsU}
	Let $G=(V,E)$ be a graph. If there exists a matching saturating $U \subseteq V$, then 
	there exists a maximum matching saturating $U$.
\end{lemma}
\vspace{-0.2cm}\begin{proof}
	Let $M$ be a matching saturating $U$. If $M$ is not maximum, one can find an
	augmenting path $P$ with respect to $M$.
	Interchanging in $P$ the edges of $M$ and $\overline{M}$
	we obtain a matching $M'$ with $|M'| > |M|$. 
	$M'$ saturates the end vertices of $P$ in addition to
	all vertices saturated by $M$. In particular, $M'$ still saturates $U$. Repeating this
	procedure, we will finally obtain a maximum matching saturating $U$.
\end{proof}

\begin{lemma}\label{lem:sat_v}
	Let $G=(V,E)$ be a graph and let $M$ be a maximum matching saturating 
	$U \subseteq V$. Let $M'$ be a maximum matching in $G$ leaving a vertex $v \in U$ unsaturated.
	Then there is an alternating path $P$ of even length with respect to $M'$ starting
	at $v$ (with an edge in $\overline{M'})$ and ending (with an edge in $M'$) at some
	vertex $w \notin U$.
\end{lemma}

\vspace{-0.2cm}\begin{proof}
	Let $H$ be a spanning subgraph of $G$ with edge set 
	$(M \setminus M') \cup (M' \setminus M)$. Since both $M$ and $M'$ are maximum,
	$H$ is the union of vertex disjoint alternating cycles and 
	alternating paths of even length. Since $v$ is saturated by $M$, but not by $M'$,
	there is at least one alternating path $P$ of even length, which starts at $v$ and ends
	at some vertex $w$ with an edge in $M'$. Since $w$ is not saturated by $M$, it is not
	in $U$.
\end{proof}

Let $G=(V,E)$ be a graph and $M'$ be a maximum matching in $G$ leaving a vertex 
$v \in U$ unsaturated. If there is an alternating path $P$ of even length with respect to $M'$ starting at $v$ and ending at some vertex $w \notin U$, then by interchanging the edges
of $P \cap M'$ and $P \cap \overline{M'}$ we obtain a matching saturating the same 
set of vertices as $M'$ except that $w \not\in U$ has been replaced by $v \in U$. Repeating this procedure we will
finally either get a matching $M^*$ saturating all vertices in $U$ or show that such a
matching cannot exist. In the latter case no alternating path of even length can be found
between an unsaturated vertex of $U$ and a saturated vertex of $\overline{U}$. This means
that all such alternating paths will have both endpoints in $U$.

\begin{lemma}\label{lem:polyMatching}
	Given a graph $G=(V,E)$ and a subset $U \subseteq V$ of vertices, one can determine
	in polynomial time whether $G$ has a matching $M$ which saturates all vertices in
	$U$ or not.
\end{lemma}

\vspace{-0.2cm}\begin{proof}
	Taking into account the above discussion, it is sufficient to show that an alternating 
	path of even length (with respect to some maximum matching $M$) between 
	an unsaturated vertex $v \in U$ and a saturated vertex $w \notin U$ can be found 
	(if it exists) in polynomial time.
	
	Let $W$ be the set of vertices saturated by $M$ which are not in $U$. We introduce
	a new vertex $u_0$ and edges $(u_0,w)$ for each vertex $w \in W$. Let $G'$ be the
	resulting graph.
	It is easy to see that $G$ has an alternating path $P$ of even length from 
	an unsaturated vertex $v \in U$ to a saturated vertex $w \notin U$ if and only if there 
	is an augmenting path $P'$ starting at $v$ in $G'$. Such a path will necessarily end 
	at vertex $u_0$, otherwise $M$ would not be maximum in $G$. Moreover, $P$ can
	be easily obtained from $P'$ by removing vertex $u_0$.
	
	To summarize, the procedure of determining whether $G$ has a maximum matching
	saturating a subset $U$ of vertices is the following.
	\begin{enumerate}
\vspace{-0.3cm}		\item[(a)] Construct a maximum matching $M$ in $G$ and let $U^*$ be the set of
		vertices in $U$ which are not saturated by the current matching $M$. 
		If $U^* = \emptyset$, then $M$ saturates $U$ and we are done.
		
\vspace{-0.3cm}		\item[(b)] As long as $U^* \neq \emptyset$, take a vertex $v \in U^*$ and try
		to find an augmenting path in $G'$ starting at $v$. If such a path $P$ can be found, 
		interchange the edges of $P \cap M$ and $P \cap \overline{M}$. Remove
		$u_0$ (and its adjacent edges) and let $M$ be the new matching and remove $v$
		from $U^*$.
		
		If no augmenting path in $G'$ starting at $v$ can be found, this means that there is no 
		alternating path of even length from $v$ to a vertex $w \notin U$ in $G$, so no
		matching $M$ saturating $U$ exists and we stop.
	\end{enumerate}
	
	By this procedure we will either show that no matching saturating $U$ exists or we 
	will obtain a matching $M$ saturating $U$. Since the procedure of finding an augmenting
	path starting at some vertex $v$ is precisely the one used in a classical maximum matching 
	algorithm and since we apply it at most $|V(G)|$ times, we have overall a polynomial
	algorithm.
\end{proof}

\begin{theorem}
The dominating induced matching problem can be solved in polynomial time in $S_{2,2,2}$-free graphs. 
Moreover, if an $S_{2,2,2}$-free graph admits a feasible complete coloring, then such a coloring can be determined in polynomial time. 
\end{theorem}

\vspace{-0.3cm}\begin{proof}
Let $G=(V,E)$ be an $S_{2,2,2}$-free graph. After having applied all forcing rules, 
propagation rules, all graph reductions and transformations, and all cleanings, three scenarios are possible. 
The first one is that we get a proof that $G$ does not admit any feasible complete coloring. This occurs if both colors black 
and white are imposed on the same vertex, or if we get a graph with a connected component isomorphic to $C_5$ (see Lemma \ref{lem:Cp+1tree}). Otherwise, we obtain an irreducible pair $(H,\gamma)$ so that $H$ is $S_{2,2,2}$-free, $H$ being possibly empty. 
We know from the lemmas of Sections \ref{sec:REDUCTION} and \ref{sec:TRANSFORMATION} that $\gamma$ is completable if and only if $G$ admits a feasible
complete coloring.

If $H$ is not empty, we create a new graph $\mathcal{G}=(\mathcal{V},\mathcal{E})$ as follows. We create a vertex $a_i\in \mathcal{V}$ 
for each subset $A_i$ in $\mathcal{K}_{(H,\gamma)}$, 
and two vertices $a_i$ and $a_j$ are adjacent if and only if the corresponding subsets $A_i$ and $A_j$ have a common element. 
Let $U$ be the subset of vertices $a_i\in \mathcal{V}$ such that each element of $A_i$ belongs to exactly two members of $\mathcal{K}_{(H,\gamma)}$.
We then know from Lemmas \ref{lem:dim_sets_reduction} and \ref{lem:A_equiv_B} that $\gamma$ is completable if and only if 
$\cal{G}$ has a matching which saturates all vertices in $U$.

Observe that all forcing and propagation rules, graph reductions, graph transformations and cleanings 
can be implemented in polynomial time. They are all applied a polynomial number of times. Indeed, define a vertex as {\it good} if 
it has less than four neighbors, or if it has degree four, belongs to a unique triangle and the two other vertices of the triangle have degree two.
The other vertices are called {\it bad}. We know from Lemma \ref{lem:insideTriangle} that graph $H$ of the irreducible pair $(H,\gamma)$ 
does not contain any bad vertex. Note now that none of the graph reductions and transformations increases the number of bad vertices.
In fact, transformations $\tau_{1},\dots,\tau_{6}$ strictly decrease the number of bad vertices.
Since $\tau_{8},\tau_{9}$ as well as reductions $\rho_1,\dots,\rho_8$ strictly decrease the number of vertices, while $\tau_{7}$ keeps 
the number of vertices constant but decreases the number of edges, we conclude that the eight graph reductions and nine graph transformations 
are used a polynomial number of times. Also, we know from Lemma \ref{lem:polyMatching} that one can determine in polynomial time whether $\cal{G}$ 
has a matching which saturates all vertices in $U$. Hence the whole procedure is polynomial 

Note finally that Lemmas \ref{lem:dim_sets_reduction} and \ref{lem:A_equiv_B} show how to obtain a feasible complete coloring of $H$ from a matching in $\mathcal{G}$ that 
saturates all vertices in $U$, while the lemmas corresponding to graph 
reductions or transformations show how to obtain a feasible complete coloring of $G$ from a feasible complete coloring of $H$.
\end{proof}


\vspace{-0.3cm}\section{Conclusion}
\label{sec:conclusion}


In this paper, we proved that the {\sc dominating induced matching} problem is polynomial-time solvable 
in the class of $S_{2,2,2}$-free graphs. This result supports our Conjecture~\ref{con:1} about complexity of the problem
in finitely defined classes. Proving (or disproving) it in 
its whole generality is a challenging task. As a step toward the complete proof, 
we suggest analyzing this conjecture under the additional restriction to triangle-free graphs of vertex degree 
at most 3. This is precisely the class of $(K_3,K_{1,4})$-free graphs, where the problem is NP-complete by 
Theorem~\ref{thm:1}. One more direction of particular importance in verifying Conjecture~\ref{con:1} is the case of $P_k$-free graphs.
To date, the solution is available only for $k=7$ \cite{P7}.

\end{document}